\documentclass[conference]{IEEEtran}
\IEEEoverridecommandlockouts
\usepackage[utf8]{inputenc}
\usepackage[utf8]{inputenc} 
\usepackage[T1]{fontenc}
\usepackage{url}
\usepackage{ifthen}
\usepackage{cite}
\usepackage[cmex10]{amsmath} 

\usepackage{cite}
\usepackage{amsmath,amssymb,amsfonts}
\usepackage{mathtools}
\usepackage{amsthm}
\usepackage{algorithmic}
\usepackage{graphicx}
\usepackage{textcomp}
\usepackage{tikz}
\usepackage{caption}
\usepackage{cuted}
\usepackage{romannum}
\usepackage[utf8]{inputenc}
\usepackage{pgfplots} 
\usepackage{pgfgantt}
\usepackage{pdflscape}
\usepackage{amssymb}
\usepackage{comment}
\usepackage{pst-plot}
\usetikzlibrary{spy}
\usetikzlibrary{positioning,calc}
\usetikzlibrary{decorations.pathmorphing,calc,shapes,shapes.geometric,patterns}
\usetikzlibrary{shapes.multipart}
\usepackage{xfrac}
\usepackage{colortbl}
\usepackage{cancel}
\usetikzlibrary{arrows,positioning,calc,intersections}
\usetikzlibrary{datavisualization.formats.functions}
\def\BibTeX{{\rm B\kern-.05em{\sc i\kern-.025em b}\kern-.08em
    T\kern-.1667em\lower.7ex\hbox{E}\kern-.125emX}}
    
\usepackage{pgfplots}
\usepgfplotslibrary{fillbetween}
\usetikzlibrary{arrows, decorations.markings}
\newtheorem{theorem}{Theorem}
\newtheorem{lemma}{Lemma}

\newtheorem{definition}{Definition}

\newcommand{\setd}{\ensuremath{\mathcal{D}}}

\newcommand{\bs}[1]{\boldsymbol{#1}}

\newcommand{\mc}{\mathcal}

\definecolor{calpolypomonagreen}{rgb}{0.12, 0.3, 0.17}
\newcounter{remarkcount}
\newenvironment{remark}{\refstepcounter{remarkcount}\begin{trivlist}\item \textbf{Remark \theremarkcount.}}{\end{trivlist}}
\newcommand{\circlearrow}{}
\DeclareRobustCommand{\circlearrow}{%
  \mathrel{\vphantom{\rightarrow}\mathpalette\circle@arrow\relax}%
}
\newcommand{\circle@arrow}[2]{%
  \m@th
  \ooalign{%
    \hidewidth$#1\circ\mkern1mu$\hidewidth\cr
    $#1-$\cr}%
}
\makeatother
\let\emptyset\varnothing
\usepackage{amsmath, amssymb, amsfonts, amsthm}
\usepackage{bm}
\usepackage{xcolor}
\usepackage{framed}

\usepackage{pgf}
\usepackage{tikz}
\usetikzlibrary{arrows,automata}
\usetikzlibrary{positioning}
\usetikzlibrary{decorations.pathmorphing}
\usetikzlibrary{shapes.geometric}
\usetikzlibrary{fit}
\usetikzlibrary{backgrounds}

\newcommand{\mbf}{\mathbf}
\newcommand{\mbb}{\mathbb}

\newcommand{\vect}{\mathrm{vec}}

\theoremstyle{definition}

\newtheorem{claim}{Claim}
\newtheorem{claimproof}{Proof of Claim}

\theoremstyle{remark}

\def\BibTeX{{\rm B\kern-.05em{\sc i\kern-.025em b}\kern-.08em
    T\kern-.1667em\lower.7ex\hbox{E}\kern-.125emX}}
\begin{document}
\onecolumn
\title{ A Single-Letter Capacity Formula for MIMO Gauss-Markov Rayleigh Fading Channels
}

\author{
\IEEEauthorblockN{Rami Ezzine\IEEEauthorrefmark{1}, Moritz Wiese\IEEEauthorrefmark{1}\IEEEauthorrefmark{3}, Christian Deppe\IEEEauthorrefmark{2}\IEEEauthorrefmark{4} and Holger Boche\IEEEauthorrefmark{1}\IEEEauthorrefmark{3}\IEEEauthorrefmark{4}}
\IEEEauthorblockA{\IEEEauthorrefmark{1}Technical University of Munich, Chair of Theoretical Information Technology, Munich, Germany\\
\IEEEauthorrefmark{2}Technical University of Munich, Institute for Communications Engineering,  Munich, Germany\\
\IEEEauthorrefmark{3}CASA -- Cyber Security in the Age of Large-Scale Adversaries–
Exzellenzcluster, Ruhr-Universit\"at Bochum, Germany\\
\IEEEauthorrefmark{4}BMBF Research Hub 6G-life, Munich, Germany\\
Email: \{rami.ezzine, wiese, christian.deppe, boche\}@tum.de}
}
\maketitle
\thispagestyle{plain}
\pagenumbering{arabic}
\pagestyle{plain}
\begin{abstract}
Over the past decades, the problem of communication over finite-state Markov channels (FSMCs) has been investigated in many researches and the capacity of FSMCs has been studied in closed form under the assumption of the availability of partial/complete channel state information at the sender and/or the receiver.   
In our work, we focus on infinite-state Markov channels by investigating the problem of message transmission over time-varying  single-user multiple-input multiple-output (MIMO) Gauss-Markov Rayleigh fading channels with average power constraint and with complete channel state information available at the receiver side (CSIR). We completely solve the problem by giving a single-letter characterization of the channel capacity in closed form and by providing a proof of it.

\end{abstract}

\section{Introduction}
In many new applications in modern wireless communications such as several machine-to-machine
and human-to-machine systems, the tactile internet \cite{tactiles} and industry 4.0 \cite{industrie40}, robust and ultra-reliable low latency information exchange is required.
These applications impose challenges on the robustness requirement because of the time-varying nature of the channel conditions caused by the mobility and the changing wireless medium.

Several accurate tractable channel models are employed to model the channel variations appearing in wireless communications including the Markov model, often employed in flat fading and inter-symbol interference \cite{Gallager}. The Markov model is widely used for modeling wireless flat-fading channels 
due to its low memory and its consolidated theory.

The availability and quality of channel state information
(CSI) has a high influence on the capacity of the Markov channels.
Over the past decades, many researchers have addressed the problem of communication over finite-state Markov channels (FSMCs) \cite{survey} and
extensive studies have been performed to analyze the capacity of FSMCs in closed form under the assumption of the availability of partial/complete channel state information at the sender and/or the receiver side\cite{Goldsmith,capacityanalysis,communicationsreceiverbased,firstordermarkov,finitestatewithfeedback,finitestaterayleighfading,markovdelayedfeedback}. 

In our work, the focus is on \color{black} discrete-time continuous-valued time-varying Markov channels\color{black}, which  are  of high relevance  for practical systems. 
In particular, we are concerned with the time-varying single-user multiple-input multiple-output (MIMO) Rayleigh fading channels, where we assume that the statistics of the gain sequence are known to both the sender and the receiver and that the actual realization of the channel state sequence is completely known to the receiver only (CSIR). Therefore, the state sequence is viewed as a second output sequence of the channel. 
We further assume that the channel fades are modeled as a first-order Gauss-Markov process, which is widely used to describe the time-varying aspect of the channel\cite{mimorayleighfading,gaussmarkov1,gaussmarkov2,gaussmarkov3}. The focus is on the  multiple-antenna setting which has drawn considerable attention
in the area of wireless communications because MIMO systems offer higher rates and more
reliability and resistance to interference, compared to single-input single-output (SISO) systems\cite{mimobenefits}.

To the best of our knowledge, no capacity formula for MIMO Gauss-Markov fading channels with CSIR is provided in the literature. A single-letter expression for the capacity is provided in \cite{telatar} in the case when the channel fades are independent and identically distributed (i.i.d.). 
Other than that, only the proof of a general formula based on the inf-information
rate for the capacity which can be generalized for arbitrary channels with abstract alphabets is provided in \cite{verduhan}.

The main contribution of our work is to give a single-letter expression of the capacity of MIMO Gauss-Markov fading channels with average power constraint and to provide a proof of it.

\quad \textit{ Paper Outline:} The rest of the paper is organized as follows. In Section \ref{sec2}, we present the channel model, provide the key definitions and the main and auxiliary results. In Section \ref{proofthm}, we provide a rigorous proof of the  capacity of  time-varying multi-antenna Rayleigh fading channels with CSIR. Section \ref{prooflemma} is devoted to deriving an upper-bound on the variance of the normalized information density between the inputs and the outputs of the time-varying MIMO Rayleigh fading channel. This auxiliary result is used in the proof of the capacity formula.
Section \ref{conclusion} contains concluding remarks and proposes
potential future research in this field. Several auxiliary lemmas are collected in the
Appendix.

\quad \textit{Notation:}  $\mathbb{C}$ denotes the set of complex numbers and $\mbb R$ denotes the set of real numbers; $H(\cdot)$ and $h(\cdot)$  correspond to the entropy and the differential entropy function, respectively; $I(\cdot;\cdot)$ denotes the mutual information between two random variables. All information
quantities are taken to base 2. Throughout the paper,  $\log$ is taken to  base 2.  The natural exponential and the natural logarithm are denoted by $\exp$ and $\ln$, respectively. For any random variables $X$ and $Y$ whose joint probability law has a density $p_{X,Y}(x,y),$ we denote their marginal probability density function by $p_{X}(x)$ and $p_{Y}(y),$ respectively,  and their conditional probability density functions by $p_{X|Y}(x|y)$ and $p_{Y|X}(y|x).$ 
For any random variables $X$, $Y$ and $Z$, we use the notation $\color{black}X \circlearrow{Y} \circlearrow{Z}\color{black}$ to indicate a Markov chain. For any random variable $X,$ $\mathrm{var}(X)$ refers to the variance of $X$ and for any random variables $X$ and $Y,$ 
$\text{cov}(X,Y)$ refers to the covariance between $X$ and $Y.$ For any set $\mc
S,$
$|\mathcal{S}|$ stands for the cardinality of $\mathcal{S}.$ For any matrix $\mbf A,$ $\mathrm{tr}(\mbf A)$ refers to the trace of $\mbf A,$ $\lVert \mbf A\rVert$ stands for the operator norm of $\mbf A$ with respect to the Euclidean norm,  $\mbf A^{H}$ stands for the standard Hermitian transpose of $\mbf A,$  $\vect\left(\mbf A\right)$ refers to the vectorization of $\mbf A,$ $\lambda_{\max}(\mbf A)$ refers to the maximum eigenvalue
of $ \mbf A$ and $\lambda_{\min}(\mbf A)$ refers to its minimum eigenvalue. For any matrix $\mbf A$ and $\mbf B,$ we use the notation $\mbf A \preceq \mbf B$ to indicate that $\mbf B-\mbf A$ is positive semi-definite.
For any vector $\bs{X},$ $\bs{X}^{T}$ refers to its transpose.
For any random matrix $\mbf A \in \mathbb{C}^{m \times n}$ with entries $\mbf A_{i,j}$ $i=1,\hdots,m,j=1,\hdots, n,$ we define
\[
\mbb E \left[\mbf A\right] = \begin{bmatrix} 
    \mbb E\left[\mbf A_{11}\right] &  \mbb E\left[\mbf A_{12}\right] & \dots \\
    \vdots & \ddots & \\
    \mbb E\left[\mbf A_{m1}\right] &        &  \mbb E\left[\mbf A_{mn}\right]
    \end{bmatrix}.
\]

For any integer $m,$ $\mathcal{Q}_{(P,m)}$ is defined to be the set of positive semi-definite Hermitian matrices which are elements of
$\mbb C^{m\times m}$ and whose trace is smaller than or equal to $P.$ 

\section{Channel Model, Definitions and  Results}
\label{sec2}
\subsection{Channel Model}
For any block-length $n$, we consider the following channel model for the time-variant fading channel $W_{\mbf G^n}$
\begin{align}
\bs{z}_{i}=\mbf G_i\bs{t}_{i}+\bs{\xi}_{i} \quad i=1\hdots n,  \label{channelmodel.}
\end{align}
where $\bs{t}^n=(\bs{t}_1,\hdots,\bs{t}_n)\in\mbb C^{N_{T} \times n}$ and $ \bs{z}^n=(\bs{z}_1,\hdots,\bs{z}_n)\in \mbb C^{N_{R} \times n}$ are channel input and output blocks, respectively, and where $N_T$ and $N_R$ refer to the number of transmit and receive antennas, respectively.

Here, $\mbf G^n=\mbf G_1\hdots \mbf G_n$, where $\mbf G_i$ models the gain for the $i^{th}$ channel use.
We consider the following model for the gain.
For $0\leq \alpha <1:$
\color{black}
\begin{align} 
 \mbf G_i=\sqrt{\alpha}\mbf G_{i-1} +\sqrt{1-\alpha}\mbf W_i, \quad i=2\hdots n.
    \label{gainmodel}
\end{align}
\color{black}
Here, $\alpha$ corresponds to the \textit{memory factor}. We assume that $\mbf G_1$ and $\mbf W_i, i=2\hdots n,$ are i.i.d. such that  $\vect{(\mbf G_1)},\vect{(\mbf W_i)}, i=2\hdots n$ are drawn from $\mc N_{\mbb C}\left(\bs{0}_{N_RN_T},\mbf K\right),$ where $\mbf K$ is any $N_RN_T\times N_RN_T$ covariance matrix.
Therefore, the sequence of $\mbf G_i, i=1\hdots n,$ forms a Markov chain.
$\bs{\xi}^n=(\bs{\xi}_1,\hdots,\bs{\xi}_n)\in \mathbb{C}^{N_{R} \times n}$ models the noise sequence.
We further assume that the $\bs{\xi}_{i}s$ are i.i.d., where $\bs{\xi}_{i} \sim \mathcal{N}_{\mathbb{C}}\left(\mbf 0_{N_{R}},\sigma^2\mbf I_{N_{R}}\right),\ i=1\hdots n,$ that $\mbf G^{n}$ and $\bs{\xi}^n$ are mutually independent. Furthermore, we assume that the random input sequence $\bs{T}^n=(\bs{T}_1,\hdots,\bs{T}_n)$ is independent of $(\mbf G^{n},\bs{\xi}^n).$ 
It is also assumed that both the sender and the receiver know the statistics of the random gain sequence $\mbf G^n$ and that only the receiver knows its actual realization (CSIR). Therefore, $\mbf G^n$ is viewed as a second output sequence of the fading channel. 

\color{black}
\subsection{Properties of the random gain sequence}
 In the following lemmas, we present some properties of the random gain in  \eqref{gainmodel}.
\begin{lemma}
\label{inductionformula}
For $0<\alpha<1$ and $i\in \{1\hdots n\},$
    \begin{align}
     \mbf G_i=\sqrt{\alpha}^{i-1}\mbf G_1+\sqrt{1-\alpha}\sum_{j=2}^{i}\sqrt{\alpha}^{i-j}\mbf W_j.
    \nonumber \end{align}
\end{lemma}
\color{black}
\begin{proof}
We will proceed by induction.
\underline{Base Case:} Clearly, the statement of the Lemma holds for $i=1$\\~\\
\underline{Inductive step:} Show that for any $k\geq 2$, if the statement of the lemma holds for $i=k$ then it holds for $i=k+1.$\\~\\
Assume that the statement of the lemma holds for $i=k,$ then we have
\begin{align}
    \mbf G_{k}=\sqrt{\alpha}^{k-1}\mbf G_1+\sqrt{1-\alpha}\sum_{j=2}^{k}\sqrt{\alpha}^{k-j}\mbf W_j. \nonumber
 \end{align}
It follows that
\begin{align}
    &\mbf G_{k+1} \nonumber \\
    &\overset{(a)}{=}\sqrt{\alpha}\mbf G_{k} +\sqrt{1-\alpha}\mbf W_{k+1} \nonumber\\
    &\overset{(b)}{=}\sqrt{\alpha}\left[ \sqrt{\alpha}^{k-1}\mbf G_1+\sqrt{1-\alpha}\sum_{j=2}^{k}\sqrt{\alpha}^{k-j}\mbf W_j \right] +\sqrt{1-\alpha}\mbf W_{k+1} \nonumber \\
    &=\sqrt{\alpha}^{k}\mbf G_1+\sqrt{1-\alpha}\sum_{j=2}^{k}\sqrt{\alpha}^{k+1-j}\mbf W_j +\sqrt{1-\alpha}\mbf W_{k+1} \nonumber\\
    &=\sqrt{\alpha}^{k}\mbf G_1+\sqrt{1-\alpha}\sum_{j=2}^{k}\sqrt{\alpha}^{k+1-j}\mbf W_j+\sqrt{1-\alpha}\sqrt{\alpha}^{k+1-(k+1)}\mbf W_{k+1} \nonumber\\
    &=\sqrt{\alpha}^{k}\mbf G_1+\sqrt{1-\alpha}\sum_{j=2}^{k+1}\sqrt{\alpha}^{k+1-j}\mbf W_j, \nonumber
\end{align}
where $(a)$ follows from \eqref{gainmodel} and $(b)$ follows from the induction assumption.
Thus, the statement of the lemma holds for $i=k+1.$\\~\\
\underline{Conclusion}: Since both the base case and the inductive step have been proved as true, by mathematical induction the statement of the lemma holds for every $i=1\hdots n$. 
\end{proof}
\color{black}
\begin{lemma}
\label{distributiongain}
$\forall i\in\{1,\hdots,n\},$ it holds that
\begin{align}
    \vect\left(\mbf G_i\right) \sim \mc N_{\mbb C}\left(\mbf 0_{N_{R}N_{T}},\mbf K   \right),
\nonumber \end{align}
where $\mbf G_i, i=1\hdots n$ is defined in \eqref{gainmodel} with $0\leq \alpha<1.$
\end{lemma}
\color{black}
\begin{proof}
Clearly, the statement of the lemma holds for $\alpha=0.$ Now, let $0<\alpha<1.$ 
The statement of the lemma is valid for $i=1$.
 Let $i\in \{2,\hdots,n\}$ be fixed arbitrarily.
Let $\mbf G'_1=\sqrt{\alpha}^{i-1}\mbf G_1$ and $\mbf W'_j=\sqrt{1-\alpha}\sqrt{\alpha}^{i-j}\mbf W_j$ for every $j\in\{2,\hdots,i\}.$
Since $\mbf G_1$ and $\mbf W_j, j=2,\hdots,n$ are  independent, it follows that $\mbf G'_1$ and $\bs W'_j, j=2,\hdots,n$ are also independent.
Since $\vect\left( \mbf G_1\right) \sim \mc N_{\mbb C}\left(\bs{0}_{N_{R}N_{T}},\mbf K\right)$ and $\vect\left(\mbf W_j\right) \sim \mc N_{\mbb C}\left(\mbf 0_{N_{R}N_{T}},\mbf K\right)$ for every $j\in\{2,\hdots i\},$ it follows that
\begin{align}
    \vect\left(\mbf G'_1\right) \sim \mc N_{\mbb C}\left(\mbf 0_{N_{R}N_{T}}, \alpha^{i-1}\mbf K  \right)
\nonumber \end{align}
and that for every $j \in \{2,\hdots,i\}$
\begin{align}
    \vect\left(\bs W'_j\right) \sim \mc N_{\mbb C}\left(\mbf 0_{N_{R} N_{T}},\left(1-\alpha\right)\alpha^{i-j} \mbf K   \right).
\nonumber \end{align}
Now, from Lemma \ref{inductionformula}, it follows that 
\begin{align}
    \mbf G_i =\mbf G'_1 +\sum_{j=2}^{i} \mbf W'_j.
\nonumber \end{align}
As a result,
\begin{align}
   \vect\left( \mbf G_i \right) \sim \mc N_{\mbb C}\left(\mbf 0_{N_{R} N_{T}},\left[\alpha^{i-1}+\left(1-\alpha\right)\sum_{j=2}^{i} \alpha^{i-j}  \right]\mbf K  \right).  
\nonumber \end{align}
For $0<\alpha<1,$ we have
\begin{align}
    \sum_{j=2}^{i} \alpha^{i-j} 
    &=\alpha^{i}\sum_{j=2}^{i}\left(\frac{1}{\alpha}\right)^{j} \nonumber \\
    &=\alpha^{i}\left(\frac{1}{\alpha}\right)^{2}\frac{1-(\frac{1}{\alpha})^{i-1}}{1-\frac{1}{\alpha}} \nonumber \\
    &=\frac{\alpha^{i}-\alpha}{\alpha^2-\alpha} \nonumber \\
    &=\frac{1-\alpha^{i-1}}{1-\alpha}.\nonumber
\nonumber \end{align}
It follows that
\begin{align}
    \alpha^{i-1}+\left(1-\alpha\right)\sum_{j=2}^{i}\alpha^{i-j}=1.
\nonumber \end{align} 
This yields
    \begin{align}
   \vect\left( \mbf G_i \right) \sim \mc N_{\mbb C}\left(\mbf 0_{N_{R} N_{T}},\mbf K \right)   \quad \forall i \in \{1,\hdots n\}.
\nonumber \end{align}
\end{proof}
\color{black}
\begin{lemma}
\label{twoindexgainformula}
Let $i_1,i_2\in \{1,\hdots, n\}.$ Assume without loss of generality that $i_1<i_2.$ consider the gain model presented in \eqref{gainmodel}. Then, for $0<\alpha<1,$  it holds that
\begin{align}
    \mbf G_{i_{2}}=\sqrt{\alpha}^{i_{2}-i_{1}}\mbf G_{i_{1}}+\sqrt{1-\alpha}\sum_{j=i_{1}+1}^{i_{2}} \sqrt{\alpha}^{i_{2}-j} \mbf W_j.
\nonumber \end{align}
\end{lemma}
\color{black}
\begin{proof}
By Lemma \ref{inductionformula}, it holds that
  \begin{align}
        \mbf G_{i_{2}}=\sqrt{\alpha}^{i_{2}-1}\mbf G_1+\sqrt{1-\alpha}\sum_{j=2}^{i_{2}}\sqrt{\alpha}^{i_{2}-j}\mbf W_j
    \nonumber  \end{align}
and that
 \begin{align}
        \mbf G_{i_{1}}=\sqrt{\alpha}^{i_{1}-1}\mbf G_1+\sqrt{1-\alpha}\sum_{j=2}^{i_{1}}\sqrt{\alpha}^{i_{1}-j}\mbf W_j.
    \nonumber \end{align}
    Thus
\begin{align}
&\mbf G_{i_{2}}-\sqrt{\alpha}^{i_{2}-i_{1}}\mbf G_{i_{1}} \nonumber \\
&=\sqrt{\alpha}^{i_{2}-1}\mbf G_1+\sqrt{1-\alpha}\sum_{j=2}^{i_{2}}\sqrt{\alpha}^{i_{2}-j}\mbf W_j-\sqrt{\alpha}^{i_{2}-i_{1}}\left[\sqrt{\alpha}^{i_{1}-1}\mbf G_1+\sqrt{1-\alpha}\sum_{j=2}^{i_{1}}\sqrt{\alpha}^{i_{1}-j}\mbf W_j\right] \nonumber \\
&=\sqrt{\alpha}^{i_{2}-1}\mbf G_1+\sqrt{1-\alpha}\sum_{j=2}^{i_{2}}\sqrt{\alpha}^{i_{2}-j}\mbf W_j-\sqrt{\alpha}^{i_{2}-1}\mbf G_1-\sqrt{1-\alpha}\sum_{j=2}^{i_{1}}\sqrt{\alpha}^{i_{2}-j}\mbf W_j \nonumber \\
&=\sqrt{1-\alpha}\sum_{j=i_{1}+1}^{i_{2}} \sqrt{\alpha}^{i_{2}-j} \mbf W_j.
\nonumber \end{align}
\end{proof}
\color{black}
\subsection{Achievable Rate and Capacity}
Next, we define an achievable rate for the channel $W_{\mbf G^n}$ and the corresponding  capacity.
For this purpose, we begin by providing the definition of a  transmission-code for $W_{\mbf G^n}.$
\begin{definition}
\label{defcode}
A transmission-code $\Gamma$ of length $n$ and size\footnote{\text{This is the same notation used in} \cite{codingtheorems}.} $\lVert \Gamma \rVert $ with average power constraint $P$ for the channel $W_{\mbf G^n}$ is a family of pairs $\left\{(\mbf t_\ell,\setd_\ell^{(\mbf g^n)}):\mbf g^n \in \left(\mbb C^{N_R\times N_T}\right)^{n}, \quad \ell=1,\ldots,\lVert \Gamma \rVert \right\}$ such that for all $\ell,j \in \{1,\ldots,\lVert \Gamma \rVert\}$ and all $\mbf g^n$ for which $\mbf g^n \in \left(\mbb C^{N_R\times N_T}\right)^{n},$ we have:
\begin{align}
& \mbf t_\ell \in \mbb C^{N_{T}\times n},\quad \setd_\ell^{(\mbf g^n)} \subset \mbb C^{N_{R} \times n}, \nonumber \\
&\frac{1}{n}\sum_{i=1}^{n}\bs{t}_{\ell,i}^{H} \bs{t}_{\ell,i} \leq P \quad  \mbf t_\ell=(\bs{t}_{\ell,1},\hdots,\bs{t}_{\ell,n}), \label{powerconstraint} \\
&\setd_\ell^{(\mbf g^n)} \cap \setd_j^{(\mbf g^n)}= \emptyset,\quad \ell \neq j. \nonumber
\end{align}
Here, $\bs{t}_\ell, \ \ell=1,\ldots,\lVert \Gamma \rVert$  and $\setd^{(\mbf g^n)}_\ell, \  \ell=1,\ldots,\lVert \Gamma \rVert,$ are the codewords and the decoding regions, respectively.

\end{definition}
\color{black}
\begin{remark}
Since we do not assume any channel state information at the transmitter side, the codewords $\mbf{t}_\ell, \ell=1,\hdots,\lVert \Gamma \rVert,$ do not depend on the gain sequence.
\end{remark}
\begin{definition}
\label{achievrate}
 A real number $R$ is called an \textit{achievable} \textit{rate} of the channel $W_{\mbf G^n}$ if for every $\theta,\delta>0$ there exists a code sequence $(\Gamma_n)_{n=1}^\infty$, where each code $\Gamma_n$ of length $n$ is defined according to Definition \ref{defcode}, such that
    \[
        \frac{\log\lVert \Gamma_n \rVert}{n}\geq R-\delta
    \]
    and
 \[
e_{\max}(\Gamma_n)=\underset{\ell\in\{1\hdots\lVert\Gamma_n \rVert \} }{\max} \mbb E \left[W_{\mbf G^n}({\setd_\ell^{(\mbf G^n)_{c}}}|\bs{t}_\ell)\right]\leq\theta
    \]
    for sufficiently large $n$.
\end{definition}
\color{black}
\begin{definition}
The supremum of all achievable rates defined according to Definition \ref{achievrate} is called the capacity of the fading channel $W_{\mbf G^n}$ and is denoted by $C(P,N_R\times N_T)$.
\end{definition}
\subsection{Main Result}
In this section, we present the main result of our work, which is a single-letter characterization of the time-varying MIMO Gauss-Markov Rayleigh fading channel. This is illustrated in the following theorem.
\color{black}
\begin{theorem}
\label{single-letter characterization}
Let $\mbf G$ be any random matrix such that $\vect(\mbf G)\sim \mc N_{\mbb C}(\bs{0}_{N_RN_T},\mbf K).$
A single-letter characterization of the capacity of the channel in \eqref{channelmodel.} with gain model in \eqref{gainmodel} with $0\leq\alpha<1$ is
\begin{align}
  C(P,N_R\times N_T)=\underset{\mbf Q \in \mc Q_ {(P,N_T)}}{\mathrm{\max}}\mbb E \left[\log\det\left(\mathbf{I}_{N_{R}}+\frac{1}{\sigma^2}\mathbf{G}\mbf Q\mathbf{G}^{H}\right)\right].
  \label{capacityformula}
  \end{align}
\end{theorem}
\color{black}
The proof of Theorem \ref{single-letter characterization} is provided in Section \ref{proofthm}.
\begin{remark}
The capacity in \eqref{capacityformula} is the same as the capacity of i.i.d. MIMO Rayleigh fading channels. Under the CSIR assumption, the memory factor $\alpha$ has no influence on the capacity. The CSIR assumption is critical in deriving the closed-form expression for the capacity. In the absence of CSIR, no single-letter formula for the capacity exists even for SISO finite-state Markov fading channels\cite{simulationbased}\cite{optinfrate}. 
\end{remark}
\subsection{Auxiliary Result}
For the proof of Theorem \ref{single-letter characterization}, we require the following auxiliary result on the normalized information density of $W_{\mbf G^n}.$
\color{black}
\begin{lemma}
\label{boundvariance}
Let $\bs{T}^{n}=\left(\bs{T}_1,\hdots,\bs{T}_n\right)$ be an $n$-length input sequence of the channel $W_{\mbf G^n}$ in \eqref{channelmodel.} with gain model in \eqref{gainmodel} such that $0\geq \alpha<1$ and such that the $\bs{T}_{i}s$ are i.i.d., where $\bs{T}_i \sim \mc N\left(\bs{0}_{N_{T}},\tilde{\mbf Q}\right),\ i=1\hdots n,$ and $\tilde{\mbf Q} \in \mc Q_{(P,N_T)}.$ Let $\bs{Z}^{n}=\left(\bs{Z}_1,\hdots,\bs{Z}_n\right)$  be the corresponding  output sequence. Then, it holds that
\begin{align}
\mathrm{var}\left(\frac{i\left(\bs{T}^n;\bs{Z}^n,\mbf G^n\right)}{n} \right)\leq \kappa(n,\alpha),
\nonumber \end{align}
where  
\[
    \kappa(n,\alpha) =
        \begin{cases}
      \frac{c''}{n}, & \text{for } \alpha=0 \\
        \frac{2c'}{n(1-\sqrt{\alpha})}+ \frac{c''}{n}, & \text{for } 0<\alpha<1
\end{cases}
  \]
for some $c',c''>0$ and where  $\underset{n\rightarrow \infty}{\lim}\kappa(n,\alpha)=0$ for any $0\leq \alpha <1.$
\end{lemma}
\color{black}
The proof of Lemma \ref{boundvariance} is provided in Section \ref{prooflemma}.
\section{Proof of Theorem \ref{single-letter characterization}}
\label{proofthm}
\subsection{Direct Proof}
\label{direct} Let
\color{black}
  $$R_{\max}=\underset{\mbf Q \in \mathcal{Q}_{(P,N_T)}}{\max}\mbb E \left[\log\det\left(\mathbf{I}_{N_{R}}+\frac{1}{\sigma^2}\mathbf{G}\mbf Q\mathbf{G}^{H}\right)\right],$$
  \color{black}
  where $\mbf G \in \mbb C^{N_{R}\times N_{T}}$ is any random matrix satisfying $\vect\left(\mbf G\right)\sim \mc N_{\mbb C}\left(\bs{0}_{N_{R}N_{T}},\mbf K  \right).$
  We are going to show that
  \begin{align}
      C(P,N_R\times N_T)\geq  R_{\max}-\epsilon,
  \nonumber \end{align}
 with $\epsilon$ being an arbitrarily small positive constant.
  Let $\theta,\delta>0$ and 
  \begin{align}
  E_n=\{ \bs{t}^n=(\bs{t}_1,\hdots,\bs{t}_n)\in \mbb C^{N_{T}\times n}: \frac{1}{n}\sum_{i=1}^{n}\lVert \bs{t}_i \rVert^2\leq P\}.\label{seten}
  \end{align}
  We define for any $\mbf Q \in \mc Q_{(P,N_T)},$
  \begin{align}
      \phi(\mbf Q)=\mbb E \left[\log\det\left(\mathbf{I}_{N_{R}}+\frac{1}{\sigma^2}\mathbf{G}\mbf Q\mathbf{G}^{H}\right)\right]. \nonumber
  \end{align}
  Now notice that any  $\mbf Q \in   \mathcal{Q}_{(P,N_T)},$ we have
  \begin{align}
   &\log\det\left(\mathbf{I}_{N_{R}}+\frac{1}{\sigma^2}\mathbf{G}\mbf Q\mathbf{G}^{H}\right) \nonumber \\
      &\overset{(a)}{\leq} \log\det\left(\mathbf{I}_{N_{R}}+\frac{1}{\sigma^2}\lVert\mathbf{G}\mbf Q\mathbf{G}^{H}\rVert \mbf I_{N_{R}}\right)\nonumber \\
      &=\log\det\left(\left[1+\frac{1}{\sigma^2}\lVert\mathbf{G}\mbf Q\mathbf{G}^{H}\rVert\right] \mbf I_{N_{R}}\right) \nonumber \\
      &=N_R\log(1+\frac{1}{\sigma^2}\lVert\mathbf{G}\mbf Q\mathbf{G}^{H}\rVert)\nonumber \\
      &\leq \frac{N_{R}}{\ln(2)\sigma^2}\lVert\mathbf{G}\mbf Q\mathbf{G}^{H}\rVert \nonumber \\
      &\leq \frac{N_{R}}{\ln(2)\sigma^2} \lVert \mbf Q \rVert\lVert\mathbf{G}\rVert^2 \nonumber \\
      &\overset{(b)}{\leq} \frac{PN_{R}}{\ln(2)\sigma^2} \lVert\mathbf{G}\rVert^2, \nonumber
  \end{align}
  where \color{black} $(a)$ follows because $\mbf A \preceq \lVert \mbf A \rVert\mbf I_{n}$ for any Hermitian $\mbf A \in \mbb C^{n\times n}$ (by Lemma \ref{normposdefinite} in the Appendix) \color{black} and $(b)$ follows because $\lVert \mbf Q \rVert=\lambda_{\max}(\mbf Q)\leq \text{tr}(\mbf Q)\leq P.$
  Now, it holds that $\mbb E\left[\frac{PN_{R}}{\ln(2)\sigma^2} \lVert\mathbf{G}\rVert^2\right]<\infty$ since  $\mbb E\left[ \lVert \mbf G\rVert^2\right]<\infty$ (by Lemma \ref{boundedmatrixnorm} in the Appendix). 
  Therefore, it follows from the dominated convergence theorem that $\phi$ is continuous on the compact set $\mc Q_{(P,N_T)}.$
  Therefore, \color{black} one can find a  $\tilde{\mbf Q} \in\mc Q_{(P,N_T)}$ such that $\mathrm{tr}(\tilde{\mbf Q})=P-\beta$ for some $\beta>0$ and such that 
  \begin{align}
  \phi(\tilde{\mbf Q})\geq R_{\max}-\epsilon.
  \label{choiceoftildeQ}
\end{align}
\color{black}
We define $$\hat{P}=P-\beta$$ and
\begin{align}\hat{\beta}=\frac{\beta}{\ln(2)\hat{P}}-\log(1+\frac{\beta}{\hat{P}})>0. \label{hatbeta}
\end{align} 
Let us now introduce the following well-known lemma: 
\color{black}
\begin{lemma}{(Feinstein's Lemma with input constraints)\cite{feinstein}}
\label{feinsteinlemma}
Let $n>0$ be fixed arbitrarily. Consider any channel with random input sequence $T^n$, with corresponding random channel output sequence $Z^n$ and with information density $i(T^n;Z^n).$
Then, for any integer $\tau >0$, real number $\gamma>0$, and measurable set $E_n$, there exists a code with cardinality $\tau$, maximum error probability $\epsilon_n$ and block-length $n$, whose codewords are contained in the set $E_n,$ where $\epsilon_n$ satisfies
$$\epsilon_n\leq \mbb P \left[\frac{1}{n}i(T^n;Z^n)\leq \frac{\log\tau}{n}+\gamma  \right]+\mbb P\left[T^n\notin E_n  \right] +2^{-n\gamma}.       $$
\end{lemma}
\color{black}
Let $\bs{T}^n=(\bs{T}_1,\hdots,\bs{T}_n) \in \mbb C^{N_{T} \times n}$ to be the random input sequence of the channel $W_{\mbf G^n},$ where the $\bs{T}_is$ are i.i.d. such that   $\bs{T}_i \sim \mc N_{\mbb C}\left( \bs{0}_{N_{T}},\tilde{\mbf Q}\right), i=1\hdots n.$ We denote its corresponding random output sequence by $\bs{Z}^n=(\bs{Z}_1,\hdots,\bs{Z}_n).$ 
Now, we apply Lemma \ref{feinsteinlemma} for the set $E_n$ 
defined in \eqref{seten} and for $\gamma=\frac{\delta}{4}.$ 
It follows that there exists a code sequence $(\Gamma_n)_{n=1}^\infty,$ where each code $\Gamma_n$  is defined according to Definition \ref{defcode} such that
\color{black}
\begin{align}
    e_{\max}(\Gamma_n)&\leq \mbb P \left[\frac{1}{n}i(\bs{T}^n;\bs{Z}^n,\mbf G^n)  \leq \frac{1}{n}\log\lVert \Gamma_n\rVert+\frac{\delta}{4}\right] + \mbb P \left[\bs{T}^n \notin E_{n}  \right] +2^{-n\frac{\delta}{4}}, \nonumber 
 \end{align}
 where here $\mbf G^n$ is viewed as a second output sequence of $W_{\mbf G^n}$ because we assume CSIR and
\color{black}
\color{black}
where
\begin{align}
    e_{\max}(\Gamma_n)&=\underset{\ell\in\{1\hdots\lVert\Gamma_n \rVert \}}{\max} \mbb E \left[W_{\mbf G^n}({\setd_\ell^{(\mbf G^n)_c}}|\bs{t}_\ell)\right] \nonumber \\
    &=\underset{\ell\in\{1\hdots\lvert \mc M \rvert \}}{\max} \mbb E\left[\mbb P \left[\hat{M}\neq \ell|M=\ell,\mbf G^n\right]\right] \nonumber \\
    &=\underset{\ell\in\{1\hdots\lvert \mathcal{M} \rvert \}}{\max} \mbb P \left[\hat{M}\neq \ell|M=\ell\right], \nonumber
\end{align}
with $M,\hat{M}$ being the random message and the random decoded message and with $\mc M$ being the set of messages.

Choose $\lVert \Gamma_n \rVert$ such that  for sufficiently large $n$
 $$ R_{\max}-\epsilon-\delta\leq \frac{\log\lVert\Gamma_n\rVert}{n}\leq R_{\max}-\epsilon-\frac{\delta}{2}.$$
 It follows that
 \begin{align}
     \color{black} e_{\max}(\Gamma_n) \color{black}&\leq  \mbb P \left[\frac{1}{n}i(\bs{T}^n;\bs{Z}^n,\mbf G^n)  \leq R_{\max}-\epsilon-\frac{\delta}{2}\right] + \mbb P \left[\bs{T}^n \notin E_{n}  \right] +2^{-n\frac{\delta}{4}} \nonumber \\
      &\leq \color{black}\mbb P \left[\frac{1}{n}i(\bs{T}^n;\bs{Z}^n,\mbf G^n)  \leq  \phi(\tilde{\mbf Q})-\frac{\delta}{2}\right]+ \mbb P \left[\bs{T}^n \notin E_{n}  \right] +2^{-n\frac{\delta}{4}}\color{black}, \label{upperprob1}
 \end{align}
 where we used \eqref{choiceoftildeQ} in the last step.
It remains to find upper-bounds for $ \mbb P \left[\frac{1}{n}i(\bs{T}^n;\bs{Z}^n,\mbf G^n)  \leq  \phi(\tilde{\mbf Q})-\frac{\delta}{2}\right]$ and for $\mbb P \left[\bs{T}^n \notin E_{n}  \right]$ that vanish as $n$ goes to infinity.
\subsubsection{Upper-bound for $\mbb P\left[ \bs{T}^n \notin E_n\right]$}
We will prove  that $$  \color{black} \mbb P\left[ \bs{T}^n \notin E_n\right]\leq 2^{-n\hat{\beta}} \color{black},$$
where $\hat{\beta}$ is defined in \eqref{hatbeta}.
For this purpose, we will introduce and prove the following lemma:
\color{black}
\begin{lemma}
\label{probnotinconstraint}
Let $\bs{X}_i,$ $i=1, \hdots,n$ be i.i.d. $N$-dimensional complex Gaussian random vectors with mean $\bs{0}_N$ and covariance matrix $\mbf O$ whose trace is smaller than or equal to $\rho$.
Then, for any $\delta>0$
\begin{align}
    \mbb P\left[\sum_{i=1}^{n} \lVert \bs{X}_i\rVert^2\geq n(\rho+\delta)\right]\leq \left[(1+\frac{\delta}{\rho})2^{-\frac{\delta}{\ln(2)\rho}}\right]^{n}, \nonumber
\end{align}
where 
\begin{align}
    \lVert \bs{X}_i \rVert^2=\sum_{j=1}^{N} |\bs{X}_i^j|^2 \nonumber
\end{align}
and
\begin{align}
    \bs{X}_i=(\bs{X}_i^1,\hdots,\bs{X}_i^N)^{T}. \nonumber
\end{align}
\end{lemma}
\color{black}
\begin{proof}
Let $\bs{X}$ be a random vector with the same distribution as each of the $\bs{X}_i$. Then
\begin{align}
&\mbb P\left[ \sum_{i=1}^{n} \lVert \bs{X}_i\rVert^2 \geq n(\rho+\delta)      \right] \nonumber \\
&=\mbb P\left[ \sum_{i=1}^{n}\lVert \bs{X}_i\rVert^2-n(\rho+\delta)\geq 0       \right] \nonumber \\
    &\leq \mbb E \left[\exp\left(\beta\left(\sum_{i=1}^{n}\lVert \bs{X}_{i}\rVert^2-n(\rho+\delta\right)    \right)            \right]\nonumber \\
    &=\left[\exp(-[\rho+\delta]\beta)\mbb E\left[ \exp\left(\beta \lVert \bs{X} \rVert^2\right)              \right]\right]^n,
    \nonumber
\end{align}
where we used the $\bs{X}_is$ are i.i.d..
By a standard calculation which follows below, one can show that
\begin{align}
    \mbb E\left[ \exp\left(\beta \lVert \bs{X}\rVert^2\right)    \right]&=\mbb E \left[\exp\left(\beta\bs{X}^H\bs{X}\right)           \right] \nonumber \\
    &=\prod_{j=1}^{N}(1-\beta\mu_j)^{-1} \quad \beta<\beta_0, \label{standardcalculation}
\end{align}
where $\mu_1,\hdots,\mu_{N}$ are the eigenvalues of $\mbf O$, and for $\beta_0=\frac{1}{\rho}\leq \frac{1}{\mu_1+\hdots+\mu_N}\leq \underset{j\in \{1,\hdots,N\}}{\min}\frac{1}{\mu_j}$  so that all the factors are positive, whether $\mbf O$ is non-singular or singular.

To prove this, we let $r$ be the rank of $\mbf O.$ It holds that $r\leq N$. We make use of the spectral decomposition theorem to express $\mbf O$ as $\mbf S_{\mbf O}^{\star} \Lambda^{\star}{S_{\mbf O}^{\star}}^{H} $, where $\Lambda^{\star}$ is a diagonal matrix whose first $r$ diagonal elements are positive and where the remaining diagonal elements are equal to zero.
Next, we let $\mbf V^{\star}=\mbf S_{\mbf O}^{\star} {\Lambda^{\star}}^{\frac{1}{2}}$ and remove the $N-r$ last columns of $\mbf V^{\star}$, which are null vectors to obtain the matrix $\mbf V.$ Then, it can be verified that $\mbf O=\mbf V \mbf V^{H}.$
We can write $\bs{X}=\mbf  V \bs{U}^\star$
where $\bs{U}^\star \sim\mathcal{N}_{\mbb C}(\bs{0},\mbf I_{r}).$
As a result:
\begin{align}
    \bs{X}^H\bs{X}={(\bs{U}^\star)}^{H}\mbf V^{H}\mbf V \bs{U}^\star. \nonumber
\end{align}
Let $\mbf S$ be a unitary matrix which diagonalizes $\mbf V^{H} \mbf V$ such that $\mbf S^{H} \mbf V^{H} \mbf V \mbf S= \text{Diag}(\mu_1,\hdots,\mu_r)$ with $\mu_1,\hdots,\mu_r$ being the positive eigenvalues of $\mbf O=\mbf V \mbf V^{H}$ in decreasing order.
One defines $\bs{U}=\mbf S^{H} \bs{U}^{\star}.$ We have
\begin{align}
    \text{cov}(\bs{U})&=\mbf S^{H} \text{cov}(\bs{U}^{\star}) \mbf S \nonumber \\
    &=\mbf S^{H} \mbf S \nonumber \\
    &=\mbf I_{r}. \nonumber
\end{align}
Therefore, it holds that $\bs{U} \sim \mathcal{N}_{\mbb C}(\bs{0},\mbf I_r).$
Since $\mbf S$ is unitary, we have
\begin{align}
    \bs{X}^{H}\bs{X}&=\left((\mbf S^{H})^{-1}\bs{U}\right)^{H} \mbf V^{H} \mbf V (\mbf S^{H})^{-1} \bs{U} \nonumber \\
    &=\bs{U}^{H}\mbf S^{H} \mbf V^{H} \mbf V \mbf S \bs{U} \nonumber \\
    &=\bs{U}^{H}\text{Diag}(\mu_1,\hdots,\mu_r)\bs{U} \nonumber \\
    &=\sum_{j=1}^{r} \mu_j |\bs{U}_{j}|^{2}.\nonumber
\end{align}
Then, we have
\begin{align}
    \mbb E\left[ \exp\left(\beta \lVert \bs{X}\rVert^2\right)    \right] &= \mbb E \left[\prod_{j=1}^{r}\exp\left(\frac{1}{2}\beta\mu_j 2|\bs{U}_j|^2  \right)             \right] \nonumber \\
    &=\prod_{j=1}^{r} \mbb E \left[ \exp\left(\frac{1}{2}\beta\mu_j 2|\bs{U}_j|^2  \right)       \right] \nonumber \\
    &=\prod_{j=1}^{N} (1-\beta\mu_j)^{-1},             \nonumber 
\end{align}
where we used that all the $\bs{U}_j$s are independent, that $ \forall j \in \{1,\hdots,r\},  2|\bs{U}_{j}|^2$ is chi-square distributed with $k=2$ degrees of freedom and with moment generating function equal to $\mbb E\left[\exp(2t|\bs{U}_j|^2)\right]=(1-2t)^{-k/2}$ for $t<\frac{1}{2}$ and that $\forall j \in \{1,\hdots,r\}$ and for $\beta<\beta_0,$ $\frac{1}{2}\beta\mu_j<\frac{1}{2}$. This proves \eqref{standardcalculation}.

Now, it holds that
$$ \prod_{i=1}^{N}(1-\beta\mu_i)\geq 1-\beta(\mu_1+\hdots+\mu_{N})\geq 1-\beta \rho.         $$This yields
\begin{align}
\mbb P\left[ \sum_{i=1}^{n} \lVert \bs{X}_i\rVert^2 \geq n(\rho+\delta)\right]& \leq\exp(-(\rho+\delta)\beta)\mbb E\left[ \exp\left(\beta \lVert \bs{X}\rVert^2\right)    \right] \nonumber \\
&\leq \frac{\exp(-(\rho+\delta)\beta)}{1-\beta \rho}, \nonumber
\end{align}
where $0<\beta<\frac{1}{\rho}=\beta_0.$
Putting $\beta=\frac{\delta}{\rho(\delta+\rho)}<\frac{1}{\rho}$ yields
\begin{align}
   \mbb P\left[ \sum_{i=1}^{n} \lVert \bs{X}_i\rVert^2 \geq n(\rho+\delta)\right]&\leq \left(1+\frac{\delta}{\rho}\right)\exp\left(-\frac{\delta}{\rho}\right)\nonumber \\
    &=\left(1+\frac{\delta}{\rho}\right)2^{\left(-\frac{\delta}{\ln(2)\rho}\right)}.
\nonumber \end{align}
This completes the proof of Lemma \ref{probnotinconstraint} .
\end{proof}
\color{black}
By Lemma \ref{probnotinconstraint}, it holds that
\begin{align}
    \mbb P \left[\sum_{i=1}^{n}\lVert \bs{T}_i\rVert^2 \geq n(\hat{P}+\beta)  \right] &\leq \left[(1+\frac{\beta}{\hat{P}})2^{(-\frac{\beta}{\ln(2)\hat{P}})}\right]^{n} \nonumber \\
    &=2^{\left(-n\frac{\beta}{\ln(2)\hat{P}}+n\log(1+\frac{\beta}{\hat{P}})\right)} \nonumber \\
    &= 2^{-n\hat{\beta}}.\nonumber
\end{align}

As a result, we have
\begin{align}
    \mbb P\left[ \bs{T}^n \notin E_n\right]&= \mbb P \left[\sum_{i=1}^{n}\lVert \bs{T}_i\rVert^2 > nP \right]\nonumber \\
    &\leq  \mbb P \left[\sum_{i=1}^{n}\lVert \bs{T}_i\rVert^2 \geq n(\hat{P}+\beta)  \right] \nonumber \\
    &\leq 2^{-n\hat{\beta}}.\label{upperprob2}
\end{align}
\color{black}
\subsubsection{Upper-bound for  $ \mbb P \left[\frac{1}{n}i(\bs{T}^n;\bs{Z}^n,\mbf G^n)  \leq  \phi(\tilde{\mbf Q})-\frac{\delta}{2}\right]
$}
Let us introduce the following lemma:
\begin{lemma}
\label{infdensity}
 \begin{align}
 i(\bs{T}^n;\bs{Z}^n,\mbf G^n)=\sum_{i=1}^{n} i(\bs{T}_i;\bs{Z}_i,\mbf G_i).
 \nonumber \end{align}
\end{lemma}
\color{black}
\begin{proof}
We have
\begin{align}
i(\bs{T}^n;\bs{Z}^n,\mbf G^n)
&=\log\left(\frac{p_{\bs{T}^n,\bs{Z}^n,\mbf G^n}\left(\bs{T}^n,\bs{Z}^n,\mbf G^n \right)}{p_{\bs{Z}^n,\mbf G^n}\left(\bs{Z}^n,\mbf G^n\right)p_{\bs{T}^{n}}(\bs{T}^{n})}\right)\nonumber \\
&=\log\left(\frac{p_{\bs{Z}^n,\mbf G^n|\bs{T}^{n}}\left(\bs{Z}^n,\mbf G^n |\bs{T}^{n}\right)}{p_{\bs{Z}^n,\mbf G^n}\left(\bs{Z}^n,\mbf G^n\right)}\right).
\nonumber
\end{align}
Since $\mbf G^n$ and $\bs{T}^n$ are independent, we have
\begin{align}
    \log\left(\frac{p_{\bs{Z}^n,\mbf G^n|\bs{T}^{n}}\left(\bs{Z}^n,\mbf G^n |\bs{T}^{n}\right)}{p_{\bs{Z}^n,\mbf G^n}\left( \bs{Z}^n,\mbf G^n\right)}\right)=\log\left(\frac{p_{\bs{Z}^n|\mbf G^n,\bs{T}^{n}}\left(\bs{Z}^n|\mbf G^n ,\bs{T}^{n}\right)}{p_{\bs{Z}^n|\mbf G^{n}}\left( \bs{Z}^n|\mbf G^{n}\right)}\right). \nonumber
\end{align}
Furthermore, since conditioned on $(\mbf G^n, \bs{T}^n)$, the outputs are independent, we have
\begin{align}
    \log\left(\frac{p_{\bs{Z}^n|\mbf G^n,\bs{T}^{n}}\left(\bs{Z}^n|\mbf G^n ,\bs{T}^{n}\right)}{p_{\bs{Z}^n|\mbf G^{n}}\left( \bs{Z}^n|\mbf G^{n}\right)}\right)=\log\left(\frac{\prod_{i=1}^{n}p_{\bs{Z}_i|\mbf G^n,\bs{T}^n}\left(\bs{Z}_i|\mbf G^n ,\bs{T}^n\right)}{p_{\bs{Z}^n|\mbf G^{n}}\left( \bs{Z}^n|\mbf G^{n}\right)}\right)\nonumber.
\end{align}
This yields

\begin{align} 
&i(\bs{T}^n;\bs{Z}^n,\mbf G^n) \nonumber \\
&=\log\left(\frac{\prod_{i=1}^{n}p_{\bs{Z}_i|\mbf G^n,\bs{T}^n}\left(\bs{Z}_i|\mbf G^n ,\bs{T}^n\right)}{p_{\bs{Z}^n|\mbf G^{n}}\left( \bs{Z}^n|\mbf G^{n}\right)}\right) \nonumber\\
&\overset{(a)}{=}\log\left(\frac{\prod_{i=1}^{n}p_{\bs{Z}_i|\mbf G_i,\bs{T}_i}\left(\bs{Z}_i|\mbf G_i ,\bs{T}_i\right)}{p_{\bs{Z}^n|\mbf G^{n}}\left( \bs{Z}^n|\mbf G^{n}\right)}\right),\nonumber
\end{align}
where $(a)$ follows because $$ \mbf G_{1}\bs{T}_{1}\dots \mbf G_{i-1}\bs{T}_{i-1} \mbf G_{i+1}\bs{T}_{i+1}\dots \mbf G_{n}\bs{T}_{n}\bs{Z}^{i-1} \circlearrow{\mbf G_i\bs{T}_{i}}\circlearrow{\bs{Z}_{i}}$$ forms a Markov chain.

Now since conditioned on $\mbf G^n$ and for independent inputs, the outputs are independent, we have
\begin{align}
&\log\left(\frac{\prod_{i=1}^{n}p_{\bs{Z}_i|\mbf G_i,\bs{T}_i}\left(\bs{Z}_i|\mbf G_i ,\bs{T}_i\right)}{p_{\bs{Z}^n|\mbf G^{n}}\left( \bs{Z}^n|\mbf G^{n}\right)}\right) \nonumber \\
&=\log\left(\frac{\prod_{i=1}^{n}p_{\bs{Z}_i|\mbf G_i,\bs{T}_i}\left(\bs{Z}_i|\mbf G_i, \bs{T}_i\right)}{\prod_{i=1}^{n}p_{\bs{Z}_i|\mbf G^{n}}\left( \bs{Z}_i|\mbf G^{n}\right)}\right).\nonumber
\end{align}
It follows that

\begin{align}
&i(\bs{T}^n;\bs{Z}^n,\mbf G^n)\nonumber \\
&=\log\left(\frac{\prod_{i=1}^{n}p_{\bs{Z}_i|\mbf G_i,\bs{T}_i}\left(\bs{Z}_i|\mbf G_i ,\bs{T}_i\right)}{\prod_{i=1}^{n}p_{\bs{Z}_i|\mbf G^{n}}\left( \bs{Z}_i|\mbf G^{n}\right)}\right) \nonumber \\
&\overset{(a)}{=}\log\left(  \frac{ \prod_{i=1}^{n} p_{\bs{Z}_i|\mbf G_i,\bs{T}_i}\left(\bs{Z}_i|\mbf G_i ,\bs{T}_i\right)   }{\prod_{i=1}^{n}  p_{\bs{Z}_i|\mbf G_i}\left( \bs{Z}_i|\mbf G_i\right)}   \right)\nonumber \\
&=\log\left( \prod_{i=1}^{n}  \frac{ p_{\bs{Z}_i|\mbf G_i,\bs{T}_i}\left(\bs{Z}_i|\mbf G_i ,\bs{T}_i\right)   }{ p_{\bs{Z}_i|\mbf G_i}\left( \bs{Z}_i|\mbf G_i\right)}   \right)\nonumber \\
&=\sum_{i=1}^{n} \log\left(  \frac{ p_{\bs{Z}_i|\mbf G_i,\bs{T}_i}\left(\bs{Z}_i|\mbf G_i ,\bs{T}_i\right)   }{ p_{\bs{Z}_i|\mbf G_i}\left( \bs{Z}_i|\mbf G_i\right)}   \right) \nonumber \\
&=\sum_{i=1}^{n} \log\left(  \frac{ p_{\bs{Z}_i|\mbf G_i,\bs{T}_i}\left(\bs{Z}_i|\mbf G_i ,\bs{T}_i\right)p_{\mbf G_i,\bs{T}_i}(\mbf G_i,\bs{T}_i)    }{ p_{\bs{Z}_i|\mbf G_i}\left( \bs{Z}_i|\mbf G_i\right)p_{\mbf G_i,\bs{T}_i}(\mbf G_i,\bs{T}_i) }  \right) \nonumber \\
&\overset{(b)}{=}\sum_{i=1}^{n} \log\left(  \frac{ p_{\bs{Z}_i,\mbf G_i,\bs{T}_i}\left(\bs{Z}_i,\mbf G_i ,\bs{T}_i\right)   }{ p_{\bs{Z}_i,\mbf G_i}\left( \bs{Z}_i,\mbf G_i\right)p_{\bs{T}_i}(\bs{T}_i) }   \right) \nonumber \\
&=\sum_{i=1}^{n} i(\bs{T}_i;\bs{Z}_i,\mbf G_i),
\nonumber \end{align}
where
$(a)$ follows because conditioned on $\mbf G_i,$  $\bs{Z}_i$ is independent of  $\mbf G_1,\hdots,\mbf G_{i-1},\mbf G_{i+1},\hdots,\mbf G_{n}$ since $(\bs{T}_i,\bs{\xi}_i)$ is independent of $\mbf G_1,\hdots,\mbf G_{i-1},\mbf G_{i+1},\hdots,\mbf G_{n}$ and $(b)$ follows because $\bs{T}_i$ and $\mbf G_i$ are independent for $i=1\hdots n.$
\end{proof}
\color{black}
 Now, recall that we chose the inputs $\bs{T}^n$ of $W_{\mbf G^n}$ to be i.i.d such that   $\bs{T}_i \sim \mc N_{\mbb C}\left( \bs{0}_{N_{T}},\tilde{\mbf Q}\right), i=1\hdots n.$ We have using Lemma \ref{infdensity}
\begin{align}
    \color{black}\mbb E \left[\frac{1}{n}i(\bs{T}^n;\bs{Z}^n,\mbf G^n )\right] \color{black}&= \frac{1}{n}\mbb E \left[\sum_{i=1}^{n} i(\bs{T}_i;\bs{Z}_i,\mbf G_i) \right] \nonumber \\
    &=\frac{1}{n}\sum_{i=1}^{n} \mbb E \left[i(\bs{T}_i;\bs{Z}_i,\mbf G_i)  \right] \nonumber \\
    &=\frac{1}{n} \sum_{i=1}^{n}I(\bs{T}_i;\bs{Z}_i,\mbf G_i) \nonumber \\
    &=\frac{1}{n} \sum_{i=1}^{n}\left(I(\bs{T}_i;\bs{Z}_i|\mbf G_i)+I(\bs{T}_i,\mbf G_i)\right) \nonumber \\
    &=\frac{1}{n} \sum_{i=1}^{n}I(\bs{T}_i;\bs{Z}_i|\mbf G_i) \nonumber \\
    &\overset{(a)}{=}\frac{1}{n} \sum_{i=1}^{n}\mbb E \left[\log\det\left(\mathbf{I}_{N_{R}}+\frac{1}{\sigma^2}\mathbf{G}_i\tilde{\mbf Q}\mathbf{G}_i^{H}\right)\right] \nonumber \\
    &\overset{(b)}{=}\color{black}\mbb E \left[\log\det\left(\mathbf{I}_{N_{R}}+\frac{1}{\sigma^2}\mathbf{G}\tilde{\mbf Q}\mathbf{G}^{H}\right)\right] \nonumber \\
    &=\phi(\tilde{\mbf Q})\color{black},
\nonumber \end{align}
 where $(a)$ follows because  $\bs{\xi}_i \sim \mathcal{N}_{\mathbb{C}}(\bs{0}_{N_R},\sigma^2 \mathbf{I}_{N_R}), \ i=1, \hdots,n$ and because all the $\bs{T}_i's$ are i.i.d. such that   $\bs{T}_i \sim \mc N_{\mbb C}\left( \bs{0}_{N_{T}},\tilde{\mbf Q}\right), i=1\hdots n.$
and $(b)$ follows because from Lemma \ref{distributiongain}, we know that    $\vect\left(\mbf G_i\right) \sim \mc N_{\mbb C}\left(\mbf 0_{N_{R}N_{T}},\mbf I_{N_{R}N_{T}}   \right), i=1 \hdots n$ and because  $\vect\left(\mbf G\right)\sim \mc N_{\mbb C}\left(\bs{0}_{N_{R}N_{T}},\mbf I_{N_{R}N_{T}}   \right).$
It follows that

  \begin{align}
    &\color{black}\mbb P \left[\frac{1}{n}i(\bs{T}^n;\bs{Z}^n,\mbf G^n)  \leq  \phi(\tilde{\mbf Q})-\frac{\delta}{2}\right] \nonumber \\
     &= \mbb P \left[\frac{1}{n}i(\bs{T}^n;\bs{Z}^n,\mbf G^n)  \leq  \mbb E \left[\frac{1}{n}i(\bs{T}^n;\bs{Z}^n,\mbf G^n)  \right]-\frac{\delta}{2}\right] \nonumber \\
     &\leq  \mbb P \left[\Bigg\vert \frac{1}{n}i(\bs{T}^n;\bs{Z}^n,\mbf G^n) - \mbb E \left[\frac{1}{n}i(\bs{T}^n;\bs{Z}^n,\mbf G^n)  \right]\Bigg\vert \geq \frac{\delta}{2}\right] \nonumber \\
     &\overset{(a)}{\leq} \color{black}\frac{4\mathrm{var}\left(\frac{i\left(\bs{T}^n;\bs{Z}^n,\mbf G^n\right)}{n} \right)}{\delta^2} \nonumber \\
     &\overset{(b)}{\leq} \frac{4\kappa(n,\alpha)}{\delta^2}\color{black},
 \label{upperprob3} \end{align}
 where $(a)$ follows from the Chebyshev's inequality and $(b)$ follows 
because $\mathrm{var}\left(\frac{i\left(\bs{T}^n;\bs{Z}^n,\mbf G^n\right)}{n} \right)\leq \kappa(n,\alpha)$ for some $\kappa(n,\alpha)>0$ with $\underset{n\rightarrow \infty}{\lim}\kappa(n,\alpha)=0$ for any $0\leq \alpha<1$ (from the auxiliary result of Lemma \ref{boundvariance}).

From \eqref{upperprob1}, \eqref{upperprob2} and \eqref{upperprob3}, we obtain
\begin{align}
    e_{\max}(\Gamma_n) \leq  4\frac{\kappa(n)}{\delta^2}+2^{-n\hat{\beta}} +2^{-n\frac{\delta}{4}}, 
\nonumber \end{align}
where $ \underset{n\rightarrow \infty}{\lim} 4\frac{\kappa(n)}{\delta^2}+2^{-n\hat{\beta}}  +2^{-n\frac{\delta}{4}}=0.$
Therefore, for sufficiently large $n,$ it holds that
  $e_{\max}(\Gamma_n)\leq \theta$. This completes the direct proof of Theorem \ref{single-letter characterization}. 

\subsection{Converse Proof}
\label{converse}
Let $R$ be any achievable rate for the channel $W_{\mbf G^n}$ in \eqref{channelmodel.}.  So, for every $\theta,\delta>0$, there exists a code sequence $(\Gamma_n)_{n=1}^\infty$  such that 
    \[
        \frac{\log\lVert \Gamma_n\rVert}{n}\geq R-\delta
    \]
    and 
   \begin{align}
     e_{\max}(\Gamma_n)=\underset{\ell\in\{1\hdots\lVert\Gamma_n \rVert \}}{\max} \mbb E \left[W_{\mbf G^n}({\setd_\ell^{(\mbf G^n)_c}}|\bs{t}_\ell)\right]\leq\theta   \label{maxerror}
     \end{align}
    for sufficiently large $n$.

 Notice that from \eqref{maxerror}, it follows that the average error probability is also bounded from above by $\theta.$ The uniformly-distributed message $M$ is mapped to the random input sequence $\bs{T}^n=(\bs{T}_1,\hdots,\bs{T}_n)$ of the channel in \eqref{channelmodel.}, where the covariance matrix of each input $\bs{T}_i$ is denoted by $\mbf Q_i.$   Let $(\bs{Z}^n,\mbf G^n)$ the corresponding outputs, where $\bs{Z}^n=(\bs{Z}_1,\hdots,\bs{Z}_n).$ We define $\mbf Q^\star$ such that
$\mbf Q^\star=\frac{1}{n}\sum_{i=1}^{n}\mbf Q_i.$ We model the random decoded message by $\hat{M}.$ The set of messages is denoted by $\mathcal{M}.$ 
 \color{black}
 \begin{lemma}
\label{traceQ}
$$\mathrm{tr}(\mbf Q^\star)\leq P$$
\end{lemma}
\color{black}
\begin{proof}
From \eqref{powerconstraint}, it holds that
\begin{align}
    \frac{1}{n}\sum_{i=1}^{n} \bs{T}_{i}^H\bs{T}_{i}\leq P, \quad \text{almost surely}. \nonumber
\end{align}
This implies that
\begin{align}
    \mbb E\left[\frac{1}{n}\sum_{i=1}^{n} \bs{T}_{i}^H\bs{T}_{i}\right] &=\frac{1}{n}\sum_{i=1}^{n}\mbb E\left[ \bs{T}_{i}^H\bs{T}_{i}\right] \nonumber \\
    &\leq P.\nonumber
\end{align}
This yields
\begin{align}
    \mathrm{tr}\left[\mbf Q^\star\right] &= \mathrm{tr} \left[\frac{1}{n}\sum_{i=1}^{n} \mbf Q_{i}\right] \nonumber \\
    &=\frac{1}{n}\sum_{i=1}^{n}\mathrm{tr}\left[ \mbf Q_i\right] \nonumber\\
    &\leq \frac{1}{n}\sum_{i=1}^{n} \mathrm{tr}\left(\mbb E\left[ \bs{T}_{i}\bs{T}_{i}^H\right]\right) \nonumber \\
    &=\frac{1}{n}\sum_{i=1}^{n} \mbb E\left[ \mathrm{tr}\left(\bs{T}_{i}\bs{T}_{i}^H\right)\right] \nonumber \\
    &=\frac{1}{n}\sum_{i=1}^{n} \mbb E\left[ \mathrm{tr}\left(\bs{T}_{i}^H\bs{T}_{i}\right)\right] \nonumber \\
    &=\frac{1}{n}\sum_{i=1}^{n} \mbb E\left[ \bs{T}_{i}^H\bs{T}_{i}\right] \nonumber \\
    &\leq P, \nonumber
\end{align}
where we used $r=\mathrm{tr}(r)$ for scalar $r$, $\mathrm{tr}\left(\mbf A \mbf B\right)= \mathrm{tr}\left(\mbf B \mbf A\right)$ and the linearity of the expectation and of the trace operators.
\end{proof}
\color{black}
 By using $\Gamma_n$ as a transmission-code for the channel $W_{\mbf G^n}$, it follows using the fact that $M$ and $\mbf G^n$ are independent that
\begin{align}
    \mbb P\left[ \hat{M}\neq M\right]&= \mbb E \left[\mbb P\left[ M\neq \hat{M}|\mbf G^n\right]   \right] \nonumber \\
    &= \mbb E \left[ \sum_{\ell=1}^{\lvert \mc M \rvert }\mbb P [M=\ell]\mbb P\left[ \hat{M}\neq \ell |M=\ell, \mbf G^n\right]   \right] \nonumber \\
    &=\sum_{\ell=1}^{\lvert \mc M \rvert }\mbb P [M=\ell] \mbb E \left[ \mbb P\left[ \hat{M}\neq \ell |M=\ell, \mbf G^n\right]   \right] \nonumber \\
    &= \sum_{\ell=1}^{\lvert \mc M \rvert }\mbb P [M=\ell]\mbb E \left[W_{\mbf G^n}({\setd_\ell^{(\mbf G^n)_c}}|\bs{t}_\ell)\right] \nonumber \\
    &\leq e_{\max}(\Gamma_n) \nonumber \\
    &\leq  \theta. 
\nonumber \end{align}
By applying Fano's inequality, we obtain
\begin{align}
    H(M|\hat{M})&\leq 1+ \mbb P\left[M\neq \hat{M}\right] \log \lvert \mathcal{M} \rvert \nonumber \\
    &\leq 1+\theta \log \lvert \mathcal{M} \rvert \nonumber \\
    &=1+\theta H(M). \nonumber
\end{align}

Now, on the one hand, it holds that
\begin{align}
    I(M;\hat{M})&=H(M)-H(M|\hat{M}) \nonumber \\
    &\geq (1-\theta)H(M)-1. \nonumber 
\end{align}
Since
\begin{align}
    H(M) &= \log |\mathcal{M}| \nonumber \\
        &=    \log \lVert \Gamma_n \rVert \nonumber \\
        &\geq n(R-\delta),
        \nonumber
\end{align}
it follows that
\begin{align}
  n(R-\delta)\leq \frac{1+I(M;\hat{M})}{1-\theta}. \nonumber
\end{align}
On the other hand, we have
\begin{align}
  &\frac{1}{n} I(M;\hat{M}) \nonumber \\
  &\overset{(a)}{\leq} \frac{1}{n} I(\bs{T}^{n};\bs{Z}^n,\mbf G^n) \nonumber \\
  &=\frac{1}{n} I(\bs{T}^{n};\bs{Z}^n|\mbf G^n)+\frac{1}{n} I(\bs{T}^n;\mbf G^n) \nonumber \\
&\overset{(b)}{=}\frac{1}{n} I(\bs{T}^{n};\bs{Z}^n|\mbf G^n) \nonumber \\
& \overset{(c)}{=} \frac{1}{n} \sum_{i=1}^{n} I(\bs{Z}_{i};\bs{T}^{n}|\mbf G^n,\bs{Z}^{i-1}) \nonumber \\
& = \frac{1}{n} \sum_{i=1}^{n} h(\bs{Z}_{i}|\mbf G^n,\bs{Z}^{i-1})-h(\bs{Z}_{i}|\mbf G^n,\bs{T}^{n},\bs{Z}^{i-1}) \nonumber \\
& \overset{(d)}{=}\frac{1}{n} \sum_{i=1}^{n} h(\bs{Z}_{i}|\mbf G^n,\bs{Z}^{i-1})-h(\bs{Z}_{i}|\mbf G_{i},\bs{T}_{i}) \nonumber \\
& \overset{(e)}{\leq}\frac{1}{n} \sum_{i=1}^{n} h(\bs{Z}_{i}|\mbf G_i)-h(\bs{Z}_{i}|\mbf G_i,\bs{T}_{i}) \nonumber \\
& = \frac{1}{n}\sum_{i=1}^{n}  I(\bs{T}_{i};\bs{Z}_{i}|\mbf G_i) \nonumber \\
& \overset{(f)}{\leq }\frac{1}{n} \sum_{i=1}^{n} \mbb E\left[ \log\det(\mbf I_{N_{R}}+\frac{1}{\sigma^2}\mbf G_{i} \mbf Q_{i} \mbf G_{i}^{H}) \right] \nonumber \\
&\overset{(g)}{=} \frac{1}{n} \sum_{i=1}^{n} \mbb E\left[ \log\det(\mbf I_{N_{R}}+\frac{1}{\sigma^2}\mbf G \mbf Q_{i} \mbf G^{H}) \right] \nonumber \\
&= \mbb E\left[\frac{1}{n}\sum_{i=1}^{n}\log\det(\mathbf{I}_{N_{R}}+\frac{1}{\sigma^2}\mathbf{G}\mathbf{Q}_{i} \mathbf{G}^{H})    \right] \nonumber \\
&\overset{(h)}{\leq} \mbb E\left[\log\det\left(\frac{1}{n}\sum_{i=1}^{n} \left[\mathbf{I}_{N_{R}}+\frac{1}{\sigma^2}\mathbf{G}\mathbf{Q}_{i} \mathbf{G}^{H}\right]\right)\right] \nonumber \\
&=\mbb E\left[\log\det\left(\mathbf{I}_{N_{R}}+\frac{1}{\sigma^2}\mathbf{G}\left(\frac{1}{n}\sum_{i=1}^{n}\mathbf{Q}_{i}\right)\mathbf{G}^{H}\right)\right] \nonumber \\
&=\mbb E\left[\log\det\left(\mathbf{I}_{N_{R}}+\frac{1}{\sigma^2}\mathbf{G}\mbf Q^\star \mathbf{G}^{H}\right)\right] \nonumber \\
&\overset{(i)}{\leq} \underset{\mbf Q \in \mathcal{Q}_{(P,N_T)}}{\max}\mbb E \left[\log\det\left(\mathbf{I}_{N_{R}}+\frac{1}{\sigma^2}\mathbf{G}\mbf Q\mathbf{G}^{H}\right)\right], \nonumber 
\end{align}
where $(a)$ follows from the Data Processing Inequality because $M\circlearrow{\bs{T}^n}\circlearrow{\mbf G^n,\bs{Z}^n}\circlearrow{\hat{M}}$ forms a Markov chain, $(b)$ follows because $\mbf G^n$ and  $\bs{T}^n$ are independent, $(c)$ follows from the chain rule for mutual information, $(d)$ follows because 
$$ \mbf G_{1},\bs{T}_{1},\dots, \mbf G_{i-1},\bs{T}_{i-1}, \mbf G_{i+1},\bs{T}_{i+1},\dots, \mbf G_{n},\bs{T}_{n},\bs{Z}^{i-1} \circlearrow{\mbf G_i,\bs{T}_{i}}\circlearrow{\bs{Z}_{i}}$$ forms a Markov chain, $(e)$ follows because conditioning does not increase entropy, $(f)$ follows because $\bs{\xi}_{i} \sim \mathcal{N}_{\mathbb{C}}\left(\mbf 0_{N_{R}},\sigma^2\mbf I_{N_{R}}\right),\ i=1\hdots n,$
$(g)$ follows because the $\mbf G_is$ are identically distributed from Lemma \ref{distributiongain} where $\mbf G$ is a random matrix that has the same distribution as each of the $\mbf G_i$ and
$(h)$   follows from Jensen's Inequality since the function $\log\circ\det$ is concave on the set of Hermitian positive semidefinite matrices and since $\mbf I_{N_{R}}+\frac{1}{\sigma^2}\mbf G\mbf Q_{i}\mbf G^H$ is Hermitian positive semidefinite for $i=1\hdots n$, $(i)$ follows because $\mbf Q^\star=\frac{1}{n}\sum_{i=1}^{n}\mbf Q_i \in \mathcal{Q}_{(P,N_T)}$ from Lemma \ref{traceQ}. 

As a result, we have
\begin{align}
    n(R-\delta) \leq \frac{n\underset{\mbf Q \in \mathcal{Q}_{(P,N_T)}}{\max}\mbb E \left[\log\det\left(\mathbf{I}_{N_{R}}+\frac{1}{\sigma^2}\mathbf{G}\mbf Q\mathbf{G}^{H}\right)\right] +1}{1-\theta}.
\nonumber \end{align}

This implies that 
\begin{align}
    R \leq \color{black}\frac{ \underset{\mbf Q \in \mathcal{Q}_{(P,N_T)}}{\max}\mbb E \left[\log\det\left(\mathbf{I}_{N_{R}}+\frac{1}{\sigma^2}\mathbf{G}\mbf Q\mathbf{G}^{H}\right)\right]+\frac{1}{n}}{1-\theta}+\delta. \color{black}
    \label{finalinequality}
\end{align}
In particular, we can choose $\delta,\theta>0$ to be arbitrarily small such that the right-hand side of \eqref{finalinequality} is equal to $\underset{\mbf Q \in \mathcal{Q}_{(P,N_T)}}{\max}\mbb E \left[\log\det\left(\mathbf{I}_{N_{R}}+\frac{1}{\sigma^2}\mathbf{G}\mbf Q\mathbf{G}^{H}\right)\right]+\delta'$ for $n\rightarrow \infty,$ with $\delta'$ being an arbitrarily small positive constant. This completes the converse proof of Theorem \ref{single-letter characterization}.

\section{Proof of Lemma \ref{boundvariance}}
\label{prooflemma}
Let $\bs{T}^{n}=\left(\bs{T}_1,\hdots,\bs{T}_n\right)$ be an $n$-length input sequence of the channel $W_{\mbf G^n}$ such that the $\bs{T}_{i}'s$ are i.i.d., where $\bs{T}_i \sim \mc N\left(\bs{0}_{N_{T}},\tilde{\mbf Q}\right),\ i=1\hdots n$  and where $\tilde{\mbf Q} \in \mc Q_{(P,N_T)}.$  Let $\bs{Z}^n$ be the corresponding output sequence, where $\bs{Z}^n=(\bs{Z}_1,\hdots,\bs{Z}_n).$ 
By Lemma \ref{infdensity}, it holds that
\begin{align}
 i(\bs{T}^n;\bs{Z}^n,\mbf G^n)=\sum_{i=1}^{n} i(\bs{T}_i;\bs{Z}_i,\mbf G_i). \label{sumdenisities}
 \end{align}

We have
\begin{equation}
    \mathrm{var}\left(\frac{i(\bs{T}^n;\bs{Z}^n,\mbf G^n)}{n} \right)=\frac{1}{n^2} \mbb E\left[i(\bs{T}^n;\bs{Z}^n,\mbf G^n)^2\right]-\frac{1}{n^2}\mbb E \left[i(\bs{T}^n;\bs{Z}^n,\mbf G^n)\right]^2. \label{equalityvariance}
\end{equation}
Now
\begin{align}
    &\frac{1}{n^2}\mbb E\left[i(\bs{T}^n;\bs{Z}^n,\mbf G^n)^2\right] \nonumber \\ \nonumber\\
    &=\frac{1}{n^2}\mbb E\left[\left(\sum_{i=1}^{n} i(\bs{T}_i;\bs{Z}_i,\mbf G_i)\right)^2\right]\nonumber \\ \nonumber \\
    &=\frac{1}{n^2}\sum_{i=1}^{n}\sum_{k=1,k\neq i}^{n}\mbb E\left[ i(\bs{T}_i;\bs{Z}_i,\mbf G_i)i(\bs{T}_k;\bs{Z}_k,\mbf G_{k}) \right] + \frac{1}{n^2}\sum_{i=1}^{n}\mbb E\left[ i(\bs{T}_i;\bs{Z}_i,\mbf G_i)^2\right] \nonumber \\ \nonumber \\
    &=\frac{1}{n^2}\sum_{i=1}^{n}\sum_{k=1,k\neq i}^{n}\mbb E \left[\mbb E\left[ i(\bs{T}_i;\bs{Z}_i,\mbf G_i)i(\bs{T}_k;\bs{Z}_k,\mbf G_{k})|\mbf G_i, \mbf G_{k} \right]\right] + \frac{1}{n^2}\sum_{i=1}^{n}\mbb E\left[ i(\bs{T}_i;\bs{Z}_i,\mbf G_i)^2\right] \nonumber \\ \nonumber \\
    &\overset{(a)}{=}\frac{1}{n^2}\sum_{i=1}^{n}\sum_{k=1,k\neq i}^{n}\mbb E \left[\mbb E\left[ i(\bs{T}_i;\bs{Z}_i,\mbf G_i)|\mbf G_i, \mbf G_{k}\right] \mbb E \left[i(\bs{T}_k;\bs{Z}_k,\mbf G_{k})|\mbf G_i, \mbf G_{k} \right]\right] + \frac{1}{n^2}\sum_{i=1}^{n}\mbb E\left[ i(\bs{T}_i;\bs{Z}_i,\mbf G_i)^2\right] \nonumber \\ \nonumber \\
    &\overset{(b)}{=}\frac{1}{n^2}\sum_{i=1}^{n}\sum_{k=1,k\neq i}^{n}\mbb E \left[\mbb E\left[ i(\bs{T}_i;\bs{Z}_i,\mbf G_i)|\mbf G_i\right] \mbb E \left[i(\bs{T}_k;\bs{Z}_k,\mbf G_{k})| \mbf G_{k} \right]\right] + \frac{1}{n^2}\sum_{i=1}^{n}\mbb E\left[ i(\bs{T}_i;\bs{Z}_i,\mbf G_i)^2\right] \nonumber \\ \nonumber \\
    &\overset{(c)}{=}\frac{1}{n^2}\sum_{i=1}^{n}\sum_{k=1,k\neq i}^{n}\mbb E \left[\log\det(\mbf I_{N_{R}}+\frac{1}{\sigma^2}\mbf G_i \tilde{\mbf Q} \mbf G_i^{H})\log\det(\mbf I_{N_{R}}+\frac{1}{\sigma^2}\mbf G_{k} \tilde{\mbf Q} \mbf G_{k}^{H})\right]  +\frac{1}{n^2}\sum_{i=1}^{n}\mbb E\left[ i(\bs{T}_i;\bs{Z}_i,\mbf G_i)^2\right],
    \label{upperboundofmean} 
\end{align}
where $(a)$ follows because for independent inputs and conditioned on $(\mbf G_i,\mbf G_j),$ $i(\bs{T}_i;\bs{Z}_i,\mbf G_i)$ and  $i(\bs{T}_j;\bs{Z}_j,\mbf G_j)$ are independent, $(b)$ follows because for independent inputs and conditioned on $\mbf G_i,$ $i(\bs{T}_i;\bs{Z}_i,\mbf G_i)$ and $\mbf G_k$ are independent, and because for independent inputs and conditioned on $\mbf G_k,$ $i(\bs{T}_k;\bs{Z}_k,\mbf G_k)$ and $\mbf G_i$ are independent, and $(c)$ follows because $\bs{\xi}_i \sim \mathcal{N}_{\mathbb{C}}(\bs{0}_{N_R},\sigma^2 \mathbf{I}_{N_R}), i=1, \hdots,n$ and because all the $\bs{T}_is$ are i.i.d. such that   $\bs{T}_i \sim \mc N_{\mbb C}\left( \bs{0}_{N_{T}},\tilde{\mbf Q}\right), i=1\hdots n.$

It holds using \eqref{sumdenisities} that
\begin{align}
    \frac{1}{n^2}\mbb E\left[i(\bs{T}^n;\bs{Z}^n,\mbf G^n)\right]^2 
    &= \frac{1}{n^2}\mbb E\left[\sum_{i=1}^{n} i(\bs{T}_i;\bs{Z}_i,\mbf G_i)\right]^2 \nonumber \\  \nonumber \\
    &=\frac{1}{n^2}\left(  \sum_{i=1}^{n} \mbb E\left[ i(\bs{T}_i;\bs{Z}_i,\mbf G_i)\right] \right)^{2} \nonumber \\ \nonumber \\
    &\geq \frac{1}{n^2}\sum_{i=1}^{n} \sum_{k=1,k\neq i}^{n}\mbb E\left[i(\bs{T}_i;\bs{Z}_i,\mbf G_i)\right]\mbb E\left[i(\bs{T}_k;\bs{Z}_k,\mbf G_k)\right] \nonumber \\ \nonumber \\
    &=\frac{1}{n^2}\sum_{i=1}^{n} \sum_{k=1,k\neq i}^{n} \mbb E \left[\mbb E\left[i(\bs{T}_i;\bs{Z}_i,\mbf G_i)|\mbf G_i\right]\right]\mbb E \left[\mbb E\left[i(\bs{T}_k;\bs{Z}_k,\mbf G_k)|\mbf G_{k}\right]\right] \nonumber \\ \nonumber\\
    &=\frac{1}{n^2}\sum_{i=1}^{n} \sum_{k=1,k\neq i}^{n}\mbb E \left[\log\det(\mbf I_{N_{R}}+\frac{1}{\sigma^2}\mbf G_i \tilde{\mbf Q} \mbf G_i^{H})\right]\mbb E \left[ \log\det(\mbf I_{N_{R}}+\frac{1}{\sigma^2}\mbf G_{k} \tilde{\mbf Q} \mbf G_{k}^{H})\right]\nonumber \\ \nonumber \\
    &=\frac{1}{n^2}\sum_{i=1}^{n} \sum_{k=1,k\neq i}^{n}\mbb E \left[ \log\det\left(\mbf I_{N_{R}}+\frac{1}{\sigma^2}\tilde{\mbf G} \tilde{\mbf Q} \tilde{\mbf G}^H\right)\right]^2.
    \label{lowerbound2}
\end{align}

It follows from \eqref{equalityvariance}, \eqref{upperboundofmean} and \eqref{lowerbound2} that
\begin{align}
     &\mathrm{var}\left(\frac{i(\bs{T}^n;\bs{Z}^n,\mbf G^n)}{n} \right) \nonumber \\  \nonumber \\
     &\leq \frac{1}{n^2}\sum_{i=1}^{n}\sum_{k=1,k\neq i}^{n}\mbb E \left[\log\det(\mbf I_{N_{R}}+\frac{1}{\sigma^2}\mbf G_i \tilde{\mbf Q} \mbf G_i^{H})\log\det(\mbf I_{N_{R}}+\frac{1}{\sigma^2}\mbf G_{k} \tilde{\mbf Q} \mbf G_{k}^{H})\right] \nonumber \\  \nonumber \\
     & \quad +\frac{1}{n^2}\sum_{i=1}^{n}\mbb E\left[ i(\bs{T}_i;\bs{Z}_i,\mbf G_i)^2\right]-\frac{1}{n^2}\sum_{i=1}^{n} \sum_{k=1,k\neq i}^{n}\mbb E \left[ \log\det\left(\mbf I_{N_{R}}+\frac{1}{\sigma^2}\tilde{\mbf G} \tilde{\mbf Q} \tilde{\mbf G}^H\right)\right]^2 \nonumber \\  \nonumber \\
     &=\frac{1}{n^2}\sum_{i=1}^{n}\sum_{k=1,k\neq i}^{n} \left(\mbb E \left[\log\det(\mbf I_{N_{R}}+\frac{1}{\sigma^2}\mbf G_i \tilde{\mbf Q} \mbf G_i^{H})\log\det(\mbf I_{N_{R}}+\frac{1}{\sigma^2}\mbf G_{k} \tilde{\mbf Q} \mbf G_{k}^{H})\right]-\mbb E \left[ \log\det\left(\mbf I_{N_{R}}+\frac{1}{\sigma^2}\tilde{\mbf G} \tilde{\mbf Q} \tilde{\mbf G}^H\right)\right]^2\right) \nonumber \\  \nonumber \\
     &\quad + \frac{1}{n^2}\sum_{i=1}^{n}\mbb E\left[ i(\bs{T}_i;\bs{Z}_i,\mbf G_i)^2\right]. \label{lasteq}
\end{align}
By defining for any $i,k\in\{1,\hdots n\}$ with $i\neq k,$
\begin{align}
m(i,k) &=\mbb E \left[\log\det(\mbf I_{N_{R}}+\frac{1}{\sigma^2}\mbf G_i \tilde{\mbf Q} \mbf G_i^{H})\log\det(\mbf I_{N_{R}}+\frac{1}{\sigma^2}\mbf G_k \tilde{\mbf Q} \mbf G_k^{H})\right]-\mbb E \left[ \log\det\left(\mbf I_{N_{R}}+\frac{1}{\sigma^2}\tilde{\mbf G} \tilde{\mbf Q} \tilde{\mbf G}^H\right)\right]^2, \nonumber
\end{align}
we obtain using \eqref{lasteq}
\begin{align}
     &\mathrm{var}\left(\frac{i(\bs{T}^n;\bs{Z}^n,\mbf G^n)}{n} \right)\nonumber \\  \nonumber \\ &\leq
     \frac{1}{n^2}\sum_{i=1}^{n}\sum_{k=1,k\neq i}^{n} m(i,k) +\frac{1}{n^2}\sum_{i=1}^{n}\mbb E\left[ i(\bs{T}_i;\bs{Z}_i,\mbf G_i)^2\right] \nonumber \\  \nonumber \\
     &=\frac{1}{n^2}\sum_{i=1}^{n}\sum_{k=1}^{i-1} m(i,k)+\frac{1}{n^2}\sum_{i=1}^{n}\sum_{k=i+1}^{n} m(i,k)+\frac{1}{n^2}\sum_{i=1}^{n}\mbb E\left[ i(\bs{T}_i;\bs{Z}_i,\mbf G_i)^2\right]. \nonumber
\end{align}
Notice that for $\alpha=0,$ it follows from \eqref{gainmodel} that all the $\mbf G_{i}s$ are i.i.d., which implies that for any $i,k\in\{1,\hdots n\}$ with $i\neq k,$
\begin{align}
m(i,k) &=\mbb E \left[\log\det(\mbf I_{N_{R}}+\frac{1}{\sigma^2}\mbf G_i \tilde{\mbf Q} \mbf G_i^{H})\log\det(\mbf I_{N_{R}}+\frac{1}{\sigma^2}\mbf G_k \tilde{\mbf Q} \mbf G_k^{H})\right]-\mbb E \left[ \log\det\left(\mbf I_{N_{R}}+\frac{1}{\sigma^2}\tilde{\mbf G} \tilde{\mbf Q} \tilde{\mbf G}^H\right)\right]^2 \nonumber \\
&=\mbb E \left[\log\det(\mbf I_{N_{R}}+\frac{1}{\sigma^2}\mbf G_i \tilde{\mbf Q} \mbf G_i^{H})\right]\mbb E\left[\log\det(\mbf I_{N_{R}}+\frac{1}{\sigma^2}\mbf G_k \tilde{\mbf Q} \mbf G_k^{H})\right]-\mbb E \left[ \log\det\left(\mbf I_{N_{R}}+\frac{1}{\sigma^2}\tilde{\mbf G} \tilde{\mbf Q} \tilde{\mbf G}^H\right)\right]^2 \nonumber \\
&=0,\nonumber
\end{align}
where we used that $\tilde{\mbf G}$ has the same distribution as $\mbf G_i, i=1\hdots n.$
Therefore, it follows that
\begin{align}
    &\mathrm{var}\left(\frac{i(\bs{T}^n;\bs{Z}^n,\mbf G^n)}{n} \right)\nonumber \\ &\leq 
        \begin{cases}
    \frac{1}{n^2}\sum_{i=1}^{n}\mbb E\left[ i(\bs{T}_i;\bs{Z}_i,\mbf G_i)^2\right], & \text{for } \alpha=0 \\~\\
        \frac{1}{n^2}\sum_{i=1}^{n}\sum_{k=1}^{i-1} m(i,k)+\frac{1}{n^2}\sum_{i=1}^{n}\sum_{k=i+1}^{n} m(i,k)+\frac{1}{n^2}\sum_{i=1}^{n}\mbb E\left[ i(\bs{T}_i;\bs{Z}_i,\mbf G_i)^2\right] & \text{for } 0<\alpha<1.
\end{cases}
\label{sumterms}
\end{align}
Now, the goal is to find a suitable upper-bound for each term in \eqref{sumterms}.
\subsection{Upper-bound for $\frac{1}{n^2}\sum_{i=1}^{n}\sum_{k=1}^{i-1} m(i,k)+\frac{1}{n^2}\sum_{i=1}^{n}\sum_{k=i+1}^{n} m(i,k)$ for $0<\alpha<1$}
We are going to show that
\begin{align} \frac{1}{n^2}\sum_{i=1}^{n}\sum_{k=1}^{i-1} m(i,k)+\frac{1}{n^2}\sum_{i=1}^{n}\sum_{k=i+1}^{n} m(i,k)
    &\leq \frac{2c'}{n(1-\sqrt{\alpha})}, \nonumber
\end{align}
for some $c'>0.$
Let us first introduce and prove the following Lemma
\begin{lemma}
\label{meanproducti1i2}
Let $i_1,i_2\in \{1,\hdots n\}.$ Assume without loss of generality that $i_1<i_2,$ then
\begin{align}
&\mbb E \left[ \log\det(\mbf I_{N_{R}}+\frac{1}{\sigma^2}\mbf G_{i_{2}} \tilde{\mbf Q} \mbf G_{i_{2}}^{H})\log\det(\mbf I_{N_{R}}+\frac{1}{\sigma^2}\mbf G_{i_{1}} \tilde{\mbf Q} \mbf G_{i_{1}}^{H})\right] \nonumber \\
&\leq \mbb E\left[
 \log\det\left(\mbf I_{N_{R}}+\frac{1}{\sigma^2}\tilde{\mbf G}\tilde{\mbf Q} \tilde{\mbf G} ^{H}\right)\right]^2+c'\sqrt{\alpha}^{i_{2}-i_{1}}, \nonumber
\end{align}
for some $c'>0,$ where $\tilde{\mbf G}$ is a random  matrix such that $\vect(\tilde{\mbf G})\sim \mc N_{\mbb C}\left(\bs{0}_{N_RN_T},\mbf K \right).$
\end{lemma}
\begin{proof}
Let $\tilde{\mbf G}$ be any random  matrix independent of $\mbf G_1$ and $\mbf W_i,i=2,\hdots n$ such that $\vect(\tilde{\mbf G})\sim \mc N_{\mbb C}\left(\bs{0}_{N_RN_T},\mbf K \right).$
By Lemma \ref{distributiongain}, it follows that $\tilde{\mbf G}$ has the same distribution as  $\mbf G_i,i=1,\hdots n.$ Furthermore, since $\tilde{\mbf G}$ is independent of $\mbf G_1$ and $\mbf W_i,i=2,\hdots n,$ it is also independent of all the $\mbf G_is.$
 By Lemma \ref{twoindexgainformula}, we know that
\begin{align}
    \mbf G_{i_{2}}=\sqrt{\alpha}^{i_{2}-i_{1}}\mbf G_{i_{1}}+\sqrt{1-\alpha}\sum_{j=i_{1}+1}^{i_{2}} \sqrt{\alpha}^{i_{2}-j} \mbf W_j.
\nonumber \end{align}
By defining $$\mbf S=\sqrt{1-\alpha}\sum_{j=i_{1}+1}^{i_{2}} \sqrt{\alpha}^{i_{2}-j} \mbf W_j,$$
it follows that
\begin{align}
    \mbf G_{i_{2}}= \sqrt{\alpha}^{i_{2}-i_{1}}\mbf G_{i_{1}}+\mbf S. \label{newgi2}
\end{align}

Define
$$\tilde{\mbf W}=\mbf S+\sqrt{\alpha}^{i_{2}-i_{1}} \tilde{\mbf G}.$$
 Notice that $\tilde{\mbf W}$ is \textit{independent} of $\mbf G_{i_1},$ since $\mbf G_{i_{1}}$ is independent of $(\mbf S, \tilde{\mbf G}).$ Analogously to the proof of 
Lemma \ref{distributiongain}, one can show that $\vect(\tilde{\mbf W})\sim \mc N_{\mbb C}\left(\bs{0}_{N_{R} N_{T}},\mbf K \right).$
 
 First, we introduce and prove the following claims:
 \begin{claim} \label{firstclaim} \textit{It holds that}
  \begin{align}
    & \log\det(\mbf I_{N_{R}}+\frac{1}{\sigma^2}\mbf G_{i_{2}} \tilde{\mbf Q} \mbf G_{i_{2}}^{H}) \nonumber \\ \nonumber \\
    &\leq \log\det\left( \mbf I_{N_{R}}+\frac{1}{\sigma^2}\tilde{\mbf W}\tilde{\mbf Q} \tilde{\mbf W}^{H} \right) +\frac{N_R}{\ln(2)\sigma^2}\sqrt{\alpha}^{i_{2}-i_{1}} \left(P\lVert \mbf G_{i_{1}} \rVert^2+P\lVert\tilde{\mbf G}\rVert^2+2\lVert \tilde{\mbf W}\tilde{\mbf Q}\tilde{\mbf G}^H \rVert +2P\lVert\mbf G_{i_{1}}\rVert \lVert\mbf S\rVert          \right). \nonumber 
 \end{align}
 \end{claim}
  \begin{claim}
 \label{secondclaim}
 \textit{It holds that}
     \begin{align}
     \log\det(\mbf I_{N_{R}}+\frac{1}{\sigma^2}\mbf G_{i_{1}} \tilde{\mbf Q} \mbf G_{i_{1}}^{H}) \leq \frac{P N_R}{\ln(2)\sigma^2}\lVert \mbf G_{i_{1}} \rVert^{2}. \nonumber  
 \end{align}
 \end{claim}
 \begin{claimproof}
From \eqref{newgi2}, we have
\begin{align}
    &\frac{1}{\sigma^2}\mbf G_{i_{2}} \tilde{\mbf Q} \mbf G_{i_{2}}^{H} \nonumber \\  \nonumber \\
    &=\frac{1}{\sigma^2} \left[\sqrt{\alpha}^{i_{2}-i_{1}}\mbf G_{i_{1}}+\mbf S \right] \tilde{\mbf Q}  \left[\sqrt{\alpha}^{i_{2}-i_{1}}\mbf G_{i_{1}}^H+\mbf S^H   \right]    \nonumber \\ \nonumber \\ 
&=\frac{1}{\sigma^2}\alpha^{i_{2}-i_{1}}\mbf G_{i_{1}} \tilde{\mbf Q} \mbf G_{i_{1}}^{H}   +\frac{1}{\sigma^2}\sqrt{\alpha}^{i_{2}-i_{1}}\mbf G_{i_{1}} \tilde{\mbf Q}\mbf S^H+\frac{1}{\sigma^2}\sqrt{\alpha}^{i_{2}-i_{1}}\mbf S\tilde{\mbf Q} \mbf G_{i_{1}}^{H}+
\frac{1}{\sigma^2}\mbf S\tilde{\mbf Q} \mbf S^H  \nonumber \\ \nonumber \\
&=\frac{1}{\sigma^2}\left[\alpha^{i_{2}-i_{1}}\mbf G_{i_{1}} \tilde{\mbf Q} \mbf G_{i_{1}}^{H}+\mbf S\tilde{\mbf Q}\mbf S^H\right]   +\frac{1}{\sigma^2}\sqrt{\alpha}^{i_{2}-i_{1}}\mbf G_{i_{1}} \tilde{\mbf Q}\mbf S^H+\frac{1}{\sigma^2}\sqrt{\alpha}^{i_{2}-i_{1}}\mbf S\tilde{\mbf Q} \mbf G_{i_{1}}^{H}.
\label{up1} \end{align}

We will prove first that
\begin{align}
      &\frac{1}{\sigma^2}\sqrt{\alpha}^{i_{2}-i_{1}}\mbf G_{i_{1}} \tilde{\mbf Q}\mbf S^H 
 +\frac{1}{\sigma^2}\sqrt{\alpha}^{i_{2}-i_{1}}\mbf S\tilde{\mbf Q} \mbf G_{i_{1}}^{H}\nonumber \\
 &\preceq\frac{2P}{\sigma^2}\sqrt{\alpha}^{i_{2}-i_{1}}\lVert\mbf G_{i_{1}}\rVert \lVert\mbf S\rVert \mbf I_{N_R}. \label{firstpart}
\end{align}

Notice now that the matrix
  $$\frac{1}{\sigma^2}\sqrt{\alpha}^{i_{2}-i_{1}}\mbf G_{i_{1}} \tilde{\mbf Q}\mbf S^H 
 +\frac{1}{\sigma^2}\sqrt{\alpha}^{i_{2}-i_{1}}\mbf S\tilde{\mbf Q} \mbf G_{i_{1}}^{H}$$ is a  Hermitian matrix since it is equal to its Hermitian transpose. \color{black} Since $\mbf A \preceq \lVert \mbf A\rVert \mbf I_{n}$ for any Hermitian matrix $\mbf A \in \mbb C^{n\times n}$ (by Lemma \ref{normposdefinite} in the Appendix), \color{black} it follows that
 \begin{align}
     &\frac{1}{\sigma^2}\sqrt{\alpha}^{i_{2}-i_{1}}\mbf G_{i_{1}} \tilde{\mbf Q}\mbf S^H 
 +\frac{1}{\sigma^2}\sqrt{\alpha}^{i_{2}-i_{1}}\mbf S\tilde{\mbf Q} \mbf G_{i_{1}}^{H} \nonumber \\ \nonumber \\
 &\preceq \Big\lVert \frac{1}{\sigma^2}\sqrt{\alpha}^{i_{2}-i_{1}}\mbf G_{i_{1}} \tilde{\mbf Q}\mbf S^H 
 +\frac{1}{\sigma^2}\sqrt{\alpha}^{i_{2}-i_{1}}\mbf S\tilde{\mbf Q} \mbf G_{i_{1}}^{H}\Big\rVert \mbf I_{N_R} \nonumber \\ \nonumber \\
 & \preceq \left( \frac{1}{\sigma^2}\sqrt{\alpha}^{i_{2}-i_{1}}\lVert\mbf G_{i_{1}}\rVert \lVert \tilde{\mbf Q} \rVert \lVert\mbf S^H\rVert
 +\frac{1}{\sigma^2}\sqrt{\alpha}^{i_{2}-i_{1}}\Vert \mbf S\rVert\lVert\tilde{\mbf Q}\rVert \lVert\mbf G_{i_{1}}^H\rVert \right) \mbf I_{N_R} \nonumber \\ \nonumber \\
 &= \frac{2}{\sigma^2}\sqrt{\alpha}^{i_{2}-i_{1}}\lVert\mbf G_{i_{1}}\rVert \lVert \tilde{\mbf Q} \rVert \lVert\mbf S\rVert \mbf I_{N_R} \nonumber \\ \nonumber \\
 &\preceq\frac{2P}{\sigma^2}\sqrt{\alpha}^{i_{2}-i_{1}}\lVert\mbf G_{i_{1}}\rVert \lVert\mbf S\rVert \mbf I_{N_R}. \nonumber
 \end{align}
 This proves \eqref{firstpart}.

Next, we will prove that
\begin{align}
    &\frac{1}{\sigma^2}\left[\alpha^{i_{2}-i_{1}}\mbf G_{i_{1}} \tilde{\mbf Q} \mbf G_{i_{1}}^{H}+\mbf S\tilde{\mbf Q}\mbf S^H\right] \nonumber \\ \nonumber \\
    &\preceq \frac{1}{\sigma^2}\tilde{\mbf W}\tilde{\mbf Q}\tilde{\mbf W}^{H}+\frac{P}{\sigma^2}\alpha^{i_{2}-i_{1}}\left(\lVert \mbf G_{i_{1}} \rVert^2+\lVert\tilde{\mbf G}\rVert^2\right)\mbf I_{N_R}+\frac{2}{\sigma^2}\sqrt{\alpha}^{i_{2}-i_{1}} \lVert \tilde{\mbf W}\tilde{\mbf Q}\tilde{\mbf G}^H \rVert\mbf I_{N_R}. \label{proofsecondpart}
\end{align}
Since $\tilde{\mbf W}=\mbf S+\sqrt{\alpha}^{i_{2}-i_{1}} \tilde{\mbf G},$ it holds  that
\begin{align}
    &\mbf S\tilde{\mbf Q} \mbf S^H \nonumber \\ \nonumber \\
    &=\left(\mbf S+\sqrt{\alpha}^{i_{2}-i_{1}} \tilde{\mbf G}-\sqrt{\alpha}^{i_{2}-i_{1}} \tilde{\mbf G}\right) \tilde{\mbf Q} \left(\mbf S^H+
    \sqrt{\alpha}^{i_{2}-i_{1}} \tilde{\mbf G}^H-\sqrt{\alpha}^{i_{2}-i_{1}} \tilde{\mbf G}^H\right) \nonumber \\ \nonumber \\
    &=\left(\tilde{\mbf W}-\sqrt{\alpha}^{i_{2}-i_{1}} \tilde{\mbf G}\right) \tilde{\mbf Q} \left(\tilde{\mbf W}^H-\sqrt{\alpha}^{i_{2}-i_{1}} \tilde{\mbf G}^H\right) \nonumber \\ \nonumber \\
    &=\tilde{\mbf W}\tilde{\mbf Q}\tilde{\mbf W}^H-\sqrt{\alpha}^{i_{2}-i_{1}}\tilde{\mbf W}\tilde{\mbf Q}\tilde{\mbf G}^H-\sqrt{\alpha}^{i_{2}-i_{1}}\tilde{\mbf G}\tilde{\mbf Q}\tilde{\mbf W}^H+\alpha^{i_{2}-i_{1}}\tilde{\mbf G} \tilde{\mbf Q}\tilde{\mbf G}^{H}. \nonumber
\end{align}
This yields
\begin{align}
    &\frac{1}{\sigma^2}\left[\alpha^{i_{2}-i_{1}}\mbf G_{i_{1}} \tilde{\mbf Q} \mbf G_{i_{1}}^{H}+\mbf S\tilde{\mbf Q}\mbf S^H\right] \nonumber \\  \nonumber \\
    &=\frac{1}{\sigma^2}\left[\alpha^{i_{2}-i_{1}}\mbf G_{i_{1}} \tilde{\mbf Q} \mbf G_{i_{1}}^{H}+\tilde{\mbf W}\tilde{\mbf Q}\tilde{\mbf W}^H-\sqrt{\alpha}^{i_{2}-i_{1}}\tilde{\mbf W}\tilde{\mbf Q}\tilde{\mbf G}^H-\sqrt{\alpha}^{i_{2}-i_{1}}\tilde{\mbf G}\tilde{\mbf Q}\tilde{\mbf W}^H+\alpha^{i_{2}-i_{1}}\tilde{\mbf G} \tilde{\mbf Q}\tilde{\mbf G}^{H}\right]  \nonumber \\ \nonumber \\
    &=\frac{1}{\sigma^2}\tilde{\mbf W}\tilde{\mbf Q}\tilde{\mbf W}^{H}+\frac{1}{\sigma^2}\alpha^{i_{2}-i_{1}}\left[\mbf G_{i_{1}} \tilde{\mbf Q} \mbf G_{i_{1}}^{H}+\tilde{\mbf G} \tilde{\mbf Q}\tilde{\mbf G}^H\right]-\frac{1}{\sigma^2}\sqrt{\alpha}^{i_{2}-i_{1}}\tilde{\mbf W}\tilde{\mbf Q}\tilde{\mbf G}^H-\frac{1}{\sigma^2}\sqrt{\alpha}^{i_{2}-i_{1}}\tilde{\mbf G}\tilde{\mbf Q}\tilde{\mbf W}^H. 
    \nonumber
    \end{align}
 Now notice that 
$$ \frac{1}{\sigma^2}\alpha^{i_{2}-i_{1}}\left[\mbf G_{i_{1}} \tilde{\mbf Q} \mbf G_{i_{1}}^{H}+\tilde{\mbf G} \tilde{\mbf Q}\tilde{\mbf G}^H\right] $$ 
is a Hermitian matrix.
This implies that
\begin{align}
    &\frac{1}{\sigma^2}\alpha^{i_{2}-i_{1}}\left[\mbf G_{i_{1}} \tilde{\mbf Q} \mbf G_{i_{1}}^{H}+\tilde{\mbf G} \tilde{\mbf Q}\tilde{\mbf G}^H\right] \nonumber \\ \nonumber \\
    &\preceq \frac{1}{\sigma^2}\alpha^{i_{2}-i_{1}}\lVert\mbf G_{i_{1}} \tilde{\mbf Q} \mbf G_{i_{1}}^{H}+\tilde{\mbf G} \tilde{\mbf Q}\tilde{\mbf G}^H\rVert \mbf I_{N_R} \nonumber \\ \nonumber \\
    &\preceq \frac{1}{\sigma^2}\alpha^{i_{2}-i_{1}}\lVert \tilde{\mbf Q}\rVert\left(\lVert \mbf G_{i_{1}} \rVert^2+\lVert\tilde{\mbf G}\rVert^2\right)\mbf I_{N_R} \nonumber \\ \nonumber \\
    &\preceq \frac{P}{\sigma^2}\alpha^{i_{2}-i_{1}}\left(\lVert \mbf G_{i_{1}} \rVert^2+\lVert\tilde{\mbf G}\rVert^2\right)\mbf I_{N_R}. \label{part1m}
\end{align}

Notice also that $-\frac{1}{\sigma^2}\sqrt{\alpha}^{i_{2}-i_{1}}\tilde{\mbf W}\tilde{\mbf Q}\tilde{\mbf G}^H-\frac{1}{\sigma^2}\sqrt{\alpha}^{i_{2}-i_{1}}\tilde{\mbf G}\tilde{\mbf Q}\tilde{\mbf W}^H $ is a Hermitian matrix. It follows that
\begin{align}
    &-\frac{1}{\sigma^2}\sqrt{\alpha}^{i_{2}-i_{1}}\tilde{\mbf W}\tilde{\mbf Q}\tilde{\mbf G}^H-\frac{1}{\sigma^2}\sqrt{\alpha}^{i_{2}-i_{1}}\tilde{\mbf G}\tilde{\mbf Q}\tilde{\mbf W}^H \nonumber \\ \nonumber \\
    &\preceq \lVert -\frac{1}{\sigma^2}\sqrt{\alpha}^{i_{2}-i_{1}}\tilde{\mbf W}\tilde{\mbf Q}\tilde{\mbf G}^H-\frac{1}{\sigma^2}\sqrt{\alpha}^{i_{2}-i_{1}}\tilde{\mbf G}\tilde{\mbf Q}\tilde{\mbf W}^H\rVert \mbf I_{N_R} \nonumber \\ \nonumber \\
    &\preceq \frac{2}{\sigma^2}\sqrt{\alpha}^{i_{2}-i_{1}} \lVert \tilde{\mbf W}\tilde{\mbf Q}\tilde{\mbf G}^H \rVert \mbf I_{N_R}. \label{part2m}
\end{align}
As a result, we have using \eqref{part1m} and \eqref{part2m}
\begin{align}
    &\frac{1}{\sigma^2}\tilde{\mbf W}\tilde{\mbf Q}\tilde{\mbf W}^{H}+\frac{1}{\sigma^2}\alpha^{i_{2}-i_{1}}\left[\mbf G_{i_{1}} \tilde{\mbf Q} \mbf G_{i_{1}}^{H}+\tilde{\mbf G} \tilde{\mbf Q}\tilde{\mbf G}^H\right]-\frac{1}{\sigma^2}\sqrt{\alpha}^{i_{2}-i_{1}}\tilde{\mbf W}\tilde{\mbf Q}\tilde{\mbf G}^H-\frac{1}{\sigma^2}\sqrt{\alpha}^{i_{2}-i_{1}}\tilde{\mbf G}\tilde{\mbf Q}\tilde{\mbf W}^H \nonumber \\ \nonumber \\
    &\preceq \frac{1}{\sigma^2}\tilde{\mbf W}\tilde{\mbf Q}\tilde{\mbf W}^{H}+\frac{P}{\sigma^2}\alpha^{i_{2}-i_{1}}\left(\lVert \mbf G_{i_{1}} \rVert^2+\lVert\tilde{\mbf G}\rVert^2\right)\mbf I_{N_R}+\frac{2}{\sigma^2}\sqrt{\alpha}^{i_{2}-i_{1}} \lVert \tilde{\mbf W}\tilde{\mbf Q}\tilde{\mbf G}^H \rVert \mbf I_{N_R}. \nonumber
\end{align}
This proves \eqref{proofsecondpart}.
We deduce using \eqref{firstpart} and \eqref{proofsecondpart} that
\begin{align}
    &\frac{1}{\sigma^2}\left[\alpha^{i_{2}-i_{1}}\mbf G_{i_{1}} \tilde{\mbf Q} \mbf G_{i_{1}}^{H}+\mbf S\tilde{\mbf Q}\tilde{\mbf S}\right]   +\frac{1}{\sigma^2}\sqrt{\alpha}^{i_{2}-i_{1}}\mbf G_{i_{1}} \tilde{\mbf Q}\mbf S^H+\frac{1}{\sigma^2}\sqrt{\alpha}^{i_{2}-i_{1}}\mbf S\tilde{\mbf Q} \mbf G_{i_{1}}^{H} \nonumber \\ \nonumber \\
    &\preceq \frac{1}{\sigma^2}\tilde{\mbf W}\tilde{\mbf Q}\tilde{\mbf W}^{H}+\frac{P}{\sigma^2}\alpha^{i_{2}-i_{1}}\left(\lVert \mbf G_{i_{1}} \rVert^2+\lVert\tilde{\mbf G}\rVert^2\right)\mbf I_{N_R}+\frac{2}{\sigma^2}\sqrt{\alpha}^{i_{2}-i_{1}} \lVert \tilde{\mbf W}\tilde{\mbf Q}\tilde{\mbf G}^H \rVert \mbf I_{N_R}+\frac{2P}{\sigma^2}\sqrt{\alpha}^{i_{2}-i_{1}}\lVert\mbf G_{i_{1}}\rVert \lVert\mbf S\rVert \mbf I_{N_R} \nonumber \\ \nonumber \\
    &= \frac{1}{\sigma^2}\tilde{\mbf W}\tilde{\mbf Q}\tilde{\mbf W}^{H}+\frac{P}{\sigma^2}\alpha^{i_{2}-i_{1}}\left(\lVert \mbf G_{i_{1}} \rVert^2+\lVert\tilde{\mbf G}\rVert^2\right)\mbf I_{N_R}  +\frac{2}{\sigma^2}\sqrt{\alpha}^{i_{2}-i_{1}}\left(\lVert \tilde{\mbf W}\tilde{\mbf Q}\tilde{\mbf G}^H \rVert +P\lVert\mbf G_{i_{1}}\rVert \lVert\mbf S\rVert\right)\mbf I_{N_R} \nonumber \\ \nonumber \\
    &\overset{(a)}{\preceq}\frac{1}{\sigma^2}\tilde{\mbf W}\tilde{\mbf Q}\tilde{\mbf W}^{H}+\frac{P}{\sigma^2}\sqrt{\alpha}^{i_{2}-i_{1}}\left(\lVert \mbf G_{i_{1}} \rVert^2+\lVert\tilde{\mbf G}\rVert^2\right)\mbf I_{N_R}  +\frac{2}{\sigma^2}\sqrt{\alpha}^{i_{2}-i_{1}}\left(\lVert \tilde{\mbf W}\tilde{\mbf Q}\tilde{\mbf G}^H \rVert +P\lVert\mbf G_{i_{1}}\rVert \lVert\mbf S\rVert\right)\mbf I_{N_R} \nonumber \\ \nonumber \\
    &=  \frac{1}{\sigma^2}\tilde{\mbf W}\tilde{\mbf Q}\tilde{\mbf W}^{H}+\frac{1}{\sigma^2}\sqrt{\alpha}^{i_{2}-i_{1}} \left(P\lVert \mbf G_{i_{1}} \rVert^2+P\lVert\tilde{\mbf G}\rVert^2+2\lVert \tilde{\mbf W}\tilde{\mbf Q}\tilde{\mbf G}^H \rVert +2P\lVert\mbf G_{i_{1}}\rVert \lVert\mbf S\rVert          \right)\mbf I_{N_R}, \label{parta2}
\end{align}
where $(a)$ follows because $\alpha<\sqrt{\alpha}$ for $0<\alpha<1.$

Therefore, it follows from \eqref{up1} and \eqref{parta2} that
\begin{align}
&\frac{1}{\sigma^2}\mbf G_{i_{2}} \tilde{\mbf Q} \mbf G_{i_{2}}^{H} \nonumber \\  \nonumber \\
    &\preceq\frac{1}{\sigma^2}\tilde{\mbf W}\tilde{\mbf Q}\tilde{\mbf W}^{H}+\frac{1}{\sigma^2}\sqrt{\alpha}^{i_{2}-i_{1}} \left(P\lVert \mbf G_{i_{1}} \rVert^2+P\lVert\tilde{\mbf G}\rVert^2+2\lVert \tilde{\mbf W}\tilde{\mbf Q}\tilde{\mbf G}^H \rVert +2P\lVert\mbf G_{i_{1}}\rVert \lVert\mbf S\rVert          \right)\mbf I_{N_R}.\nonumber
\end{align}
This yields
\begin{align}
&\log\det(\mbf I_{N_{R}}+\frac{1}{\sigma^2}\mbf G_{i_{2}} \tilde{\mbf Q} \mbf G_{i_{2}}^{H}) \nonumber \\ \nonumber \\
&\leq \log\det\left(\mbf I_{N_R}+\frac{1}{\sigma^2}\tilde{\mbf W}\tilde{\mbf Q}\tilde{\mbf W}^{H}+\frac{1}{\sigma^2}\sqrt{\alpha}^{i_{2}-i_{1}} \left(P\lVert \mbf G_{i_{1}} \rVert^2+P\lVert\tilde{\mbf G}\rVert^2+2\lVert \tilde{\mbf W}\tilde{\mbf Q}\tilde{\mbf G}^H \rVert +2P\lVert\mbf G_{i_{1}}\rVert \lVert\mbf S\rVert          \right)\mbf I_{N_R}\right).
\label{up2}  \end{align}

Now by Lemma \ref{bounddeterminantsum} in the Appendix, we know that for any positive-definite Hermitian matrix $\mbf A \in \mbb C^{n\times n} $ with smallest eigenvalue $\lambda_{\min}(\mbf A)$  and for any positive semi-definite Hermitian matrix $\mbf B \in \mbb C^{n \times n},$  the following is satisfied:
\begin{align}
    \log\det(\mbf A+\mbf B)\leq\log\det(\mbf A)+\log\det(\mbf I_{n}+\frac{1}{\lambda_{\min}(\mbf A)}\mbf B).
\nonumber \end{align}
By applying Lemma \ref{bounddeterminantsum} in the Appendix
for 

$$\mbf A=\mbf I_{N_{R}}+\frac{1}{\sigma^2}\tilde{\mbf W}\tilde{\mbf Q} \tilde{\mbf W}^{H}     $$
and 
for 
\begin{align}
    \mbf B&=\frac{1}{\sigma^2}\sqrt{\alpha}^{i_{2}-i_{1}} \left(P\lVert \mbf G_{i_{1}} \rVert^2+P\lVert\tilde{\mbf G}\rVert^2+2\lVert \tilde{\mbf W}\tilde{\mbf Q}\tilde{\mbf G}^H \rVert +2P\lVert\mbf G_{i_{1}}\rVert \lVert\mbf S\rVert          \right)\mbf I_{N_R}, \nonumber
\end{align}
it follows from \eqref{up2} that
\begin{align}
&\log\det(\mbf I_{N_{R}}+\frac{1}{\sigma^2}\mbf G_{i_{2}} \tilde{\mbf Q} \mbf G_{i_{2}}^{H}) \nonumber \\ \nonumber \\
&\leq \log\det\left( \mbf I_{N_{R}}+\frac{1}{\sigma^2}\tilde{\mbf W}\tilde{\mbf Q} \tilde{\mbf W}^{H} \right) \nonumber \\ \nonumber \\
&\quad +\log \det \left(\mbf I_{N_{R}}+\frac{\frac{1}{\sigma^2}\sqrt{\alpha}^{i_{2}-i_{1}} \left(P\lVert \mbf G_{i_{1}} \rVert^2+P\lVert\tilde{\mbf G}\rVert^2+2\lVert \tilde{\mbf W}\tilde{\mbf Q}\tilde{\mbf G}^H \rVert +2P\lVert\mbf G_{i_{1}}\rVert \lVert\mbf S\rVert          \right)}{\lambda_{\min}\left( \mbf I_{N_{R}}+\frac{1}{\sigma^2}\tilde{\mbf W}\tilde{\mbf Q} \tilde{\mbf W}^{H}   \right)}\mbf I_{N_{R}} \right) \nonumber \\ \nonumber \\
&\overset{(a)}{\leq}\log\det\left( \mbf I_{N_{R}}+\frac{1}{\sigma^2}\tilde{\mbf W}\tilde{\mbf Q} \tilde{\mbf W}^{H} \right) \nonumber \\ \nonumber \\
&\quad +\frac{1}{\ln(2)}\mathrm{tr}\left[\frac{\frac{1}{\sigma^2}\sqrt{\alpha}^{i_{2}-i_{1}} \left(P\lVert \mbf G_{i_{1}} \rVert^2+P\lVert\tilde{\mbf G}\rVert^2+2\lVert \tilde{\mbf W}\tilde{\mbf Q}\tilde{\mbf G}^H \rVert +2P\lVert\mbf G_{i_{1}}\rVert \lVert\mbf S\rVert          \right) }{\lambda_{\min}\left( \mbf I_{N_{R}}+\frac{1}{\sigma^2}\tilde{\mbf W}\tilde{\mbf Q} \tilde{\mbf W}^{H}   \right)}\mbf I_{N_{R}} \right] \nonumber \\ \nonumber \\
&\overset{(b)}{\leq} \log\det\left( \mbf I_{N_{R}}+\frac{1}{\sigma^2}\tilde{\mbf W}\tilde{\mbf Q} \tilde{\mbf W}^{H} \right) \nonumber \\ \nonumber \\
&\quad +\frac{1}{\ln(2)}\mathrm{tr}\left[\frac{1}{\sigma^2}\sqrt{\alpha}^{i_{2}-i_{1}} \left(P\lVert \mbf G_{i_{1}} \rVert^2+P\lVert\tilde{\mbf G}\rVert^2+2\lVert \tilde{\mbf W}\tilde{\mbf Q}\tilde{\mbf G}^H \rVert +2P\lVert\mbf G_{i_{1}}\rVert \lVert\mbf S\rVert          \right)\mbf I_{N_{R}}\right]  \nonumber \\ \nonumber \\
&\overset{(c)}{=} \log\det\left( \mbf I_{N_{R}}+\frac{1}{\sigma^2}\tilde{\mbf W}\tilde{\mbf Q} \tilde{\mbf W}^{H} \right)+\frac{N_R}{\ln(2)\sigma^2}\sqrt{\alpha}^{i_{2}-i_{1}} \left(P\lVert \mbf G_{i_{1}} \rVert^2+P\lVert\tilde{\mbf G}\rVert^2+2\lVert \tilde{\mbf W}\tilde{\mbf Q}\tilde{\mbf G}^H \rVert +2P\lVert\mbf G_{i_{1}}\rVert \lVert\mbf S\rVert          \right)
,\nonumber
\end{align}
where $(a)$ follows because $  \ln\det(\mbf I_{n}+\mbf A) \leq \mathrm{tr}(\mbf A)$ for positive semi-definite $\mbf A,$ $(b)$ follows because $\lambda_{\min}\left( \mbf I_{N_{R}}+\frac{1}{\sigma^2}\tilde{\mbf W}\tilde{\mbf Q} \tilde{\mbf W}^{H}   \right)\geq 1$ and  $(c)$ follows because $\mathrm{tr}(c \mbf I_{N_R})=c N_{R}$ for any constant $c.$
This completes the proof of Claim \ref{firstclaim}.
 \end{claimproof}
  \begin{claimproof}
\color{black}Since $\mbf A \preceq \lVert \mbf A \rVert\mbf I_{n}$ for any Hermitian $\mbf A \in \mbb C^{n\times n}$  (by Lemma \ref{normposdefinite} in the Appendix) \color{black} and since the matrix $\mbf G_{i_{1}} \tilde{\mbf Q} \mbf G_{i_{1}}^{H}$ is Hermitian, we have
 \begin{align}
     \log\det(\mbf I_{N_{R}}+\frac{1}{\sigma^2}\mbf G_{i_{1}} \tilde{\mbf Q} \mbf G_{i_{1}}^{H}) &\leq \log\det(\mbf I_{N_{R}}+\frac{1}{\sigma^2}\lVert \mbf G_{i_{1}} \tilde{\mbf Q} \mbf G_{i_{1}}^{H}\rVert\mbf I_{N_{R}} ) \nonumber \\ \nonumber \\
     &\leq  \frac{1}{\ln(2)}\mathrm{tr}\left[\frac{1}{\sigma^2}\lVert \mbf G_{i_{1}} \tilde{\mbf Q} \mbf G_{i_{1}}^{H}\rVert \mbf I_{N_{R}} \right] \nonumber \\ \nonumber \\
     &\leq \frac{1}{\ln(2)}\mathrm{tr}\left[\frac{1}{\sigma^2}\lVert \mbf G_{i_{1}} \rVert^{2} \lVert \tilde{\mbf Q} \rVert \mbf I_{N_{R}}\right]\nonumber \\ \nonumber \\
     &=\frac{N_R}{\ln(2)\sigma^2}\lVert \mbf G_{i_{1}} \rVert^{2} \lVert \tilde{\mbf Q} \rVert \nonumber \\ \nonumber \\
     &\leq \frac{P N_R}{\ln(2)\sigma^2}\lVert \mbf G_{i_{1}} \rVert^{2}.  \nonumber
 \end{align}
 This completes the proof of Claim \ref{secondclaim}.
   \end{claimproof}

 Now that we proved the two claims, we let 
\begin{align}
    &\Lambda\left( \mbf G_{i_{1}},\mbf S,\tilde{\mbf G},\tilde{\mbf W}\right)=\lVert \mbf G_{i_{1}} \rVert^{2} \left(P\lVert \mbf G_{i_{1}} \rVert^2+P\lVert\tilde{\mbf G}\rVert^2+2\lVert \tilde{\mbf W}\tilde{\mbf Q}\tilde{\mbf G}^H \rVert +2P\lVert\mbf G_{i_{1}}\rVert \lVert\mbf S\rVert          \right). \nonumber \\ \nonumber 
\end{align}
It follows using Claim \ref{firstclaim} and Claim \ref{secondclaim} that
\begin{align}
 &\mbb E \left[ \log\det(\mbf I_{N_{R}}+\frac{1}{\sigma^2}\mbf G_{i_{2}} \tilde{\mbf Q} \mbf G_{i_{2}}^{H})\log\det(\mbf I_{N_{R}}+\frac{1}{\sigma^2}\mbf G_{i_{1}} \tilde{\mbf Q} \mbf G_{i_{1}}^{H})\right] \nonumber \\ \nonumber \\
 &\leq \mbb E \left[\log\det\left( \mbf I_{N_{R}}+\frac{1}{\sigma^2}\tilde{\mbf W}\tilde{\mbf Q} \tilde{\mbf W}^{H} \right)
 \log\det\left(\mbf I_{N_{R}}+\frac{1}{\sigma^2}\mbf G_{i_{1}} \tilde{\mbf Q} \mbf G_{i_{1}}^{H}\right)\right]\nonumber \\ \nonumber \\
 &\quad+\mbb E\left[\frac{N_R}{\ln(2)\sigma^2}\sqrt{\alpha}^{i_{2}-i_{1}} \left(P\lVert \mbf G_{i_{1}} \rVert^2+P\lVert\tilde{\mbf G}\rVert^2+2\lVert \tilde{\mbf W}\tilde{\mbf Q}\tilde{\mbf G}^H \rVert +2P\lVert\mbf G_{i_{1}}\rVert \lVert\mbf S\rVert          \right)\frac{P N_R}{\ln(2)\sigma^2}\lVert \mbf G_{i_{1}} \rVert^{2} \right]  \nonumber \\ \nonumber \\
 &=\mbb E \left[\log\det\left( \mbf I_{N_{R}}+\frac{1}{\sigma^2}\tilde{\mbf W}\tilde{\mbf Q} \tilde{\mbf W}^{H} \right)
 \log\det\left(\mbf I_{N_{R}}+\frac{1}{\sigma^2}\mbf G_{i_{1}} \tilde{\mbf Q} \mbf G_{i_{1}}^{H}\right)\right]\nonumber \\ \nonumber \\
 &\quad+\mbb E\left[\frac{P N_R^2}{\ln(2)^2\sigma^4}\sqrt{\alpha}^{i_{2}-i_{1}}\lVert \mbf G_{i_{1}} \rVert^{2} \left(P\lVert \mbf G_{i_{1}} \rVert^2+P\lVert\tilde{\mbf G}\rVert^2+2\lVert \tilde{\mbf W}\tilde{\mbf Q}\tilde{\mbf G}^H \rVert +2P\lVert\mbf G_{i_{1}}\rVert \lVert\mbf S\rVert          \right) \right]  \nonumber \\ \nonumber \\
 &\overset{(a)}{=}\mbb E \left[\log\det\left( \mbf I_{N_{R}}+\frac{1}{\sigma^2}\tilde{\mbf W}\tilde{\mbf Q} \tilde{\mbf W}^{H} \right)\right]\mbb E\left[
 \log\det\left(\mbf I_{N_{R}}+\frac{1}{\sigma^2}\mbf G_{i_{1}} \tilde{\mbf Q} \mbf G_{i_{1}}^{H}\right)\right]+ \frac{P N_R^2}{\ln(2)^2\sigma^4}\sqrt{\alpha}^{i_{2}-i_{1}} \mbb E\left[\Lambda\left( \mbf G_{i_{1}},\mbf S,\tilde{\mbf G},\tilde{\mbf W}\right)\right] \nonumber \\ \nonumber \\
 &\overset{(b)}{=}\mbb E\left[
 \log\det\left(\mbf I_{N_{R}}+\frac{1}{\sigma^2}\tilde{\mbf G}\tilde{\mbf Q} \tilde{\mbf G} ^{H}\right)\right]^2+\frac{P N_R^2}{\ln(2)^2\sigma^4}\sqrt{\alpha}^{i_{2}-i_{1}}  \mbb E\left[\Lambda\left( \mbf G_{i_{1}},\mbf S,\tilde{\mbf G},\tilde{\mbf W}\right)\right], 
 \nonumber
\end{align}
where $(a)$ follows because $\tilde{\mbf W}$ and $\mbf G_{i_{1}}$ are independent and $(b)$ follows because $\tilde{\mbf G}$ has the same distribution as $\tilde{\mbf W}$ and $\mbf G_{i_{1}}$ since $   \vect\left(\tilde{\mbf W}\right)\sim \mc N_{\mbb C}\left( \bs{0}_{N_RN_T},\mbf K\right)$
and since from Lemma \ref{distributiongain}, we know that $   \vect\left(\mbf G_{i_{1}}\right)\sim \mc N_{\mbb C}\left( \bs{0}_{N_RN_T},\mbf K\right).$ 

Now, from Lemma \ref{meanbounded} in the Appendix, we know that $ \mbb E\left[\Lambda\left( \mbf G_{i_{1}},\mbf S,\tilde{\mbf G},\tilde{\mbf W}\right)\right]$ is bounded from above by some $c>0$. Therefore it follows that for $i_1<i_2$
\begin{align}
&\mbb E \left[ \log\det(\mbf I_{N_{R}}+\frac{1}{\sigma^2}\mbf G_{i_{2}} \tilde{\mbf Q} \mbf G_{i_{2}}^{H})\log\det(\mbf I_{N_{R}}+\frac{1}{\sigma^2}\mbf G_{i_{1}} \tilde{\mbf Q} \mbf G_{i_{1}}^{H})\right] \nonumber \\
&\leq \mbb E\left[
 \log\det\left(\mbf I_{N_{R}}+\frac{1}{\sigma^2}\tilde{\mbf G}\tilde{\mbf Q} \tilde{\mbf G} ^{H}\right)\right]^2+\frac{PN_{R}^2}{\ln(2)^2\sigma^4} c\sqrt{\alpha}^{i_{2}-i_{1}}
 \nonumber \\
 &=\mbb E\left[
 \log\det\left(\mbf I_{N_{R}}+\frac{1}{\sigma^2}\tilde{\mbf G}\tilde{\mbf Q} \tilde{\mbf G} ^{H}\right)\right]^2+c'\sqrt{\alpha}^{i_{2}-i_{1}}, \nonumber
\end{align}
for some $c>0,$ where $c'= \frac{PN_{R}^2 c}{\ln(2)^2\sigma^4}>0.$ 
This completes the proof of Lemma \ref{meanproducti1i2}.
 \end{proof}
Now that we proved Lemma \ref{meanproducti1i2}, we will use that lemma to prove that \begin{align} \frac{1}{n^2}\sum_{i=1}^{n}\sum_{k=1}^{i-1} m(i,k)+\frac{1}{n^2}\sum_{i=1}^{n}\sum_{k=i+1}^{n} m(i,k)
    &\leq \frac{2c'}{n(1-\sqrt{\alpha})}. \nonumber
\end{align}

We recall that for any $i,k\in\{1,\hdots n\}$ with $i\neq k,$
 \begin{align}
m(i,k) &=\mbb E \left[\log\det(\mbf I_{N_{R}}+\frac{1}{\sigma^2}\mbf G_i \tilde{\mbf Q} \mbf G_i^{H})\log\det(\mbf I_{N_{R}}+\frac{1}{\sigma^2}\mbf G_k \tilde{\mbf Q} \mbf G_k^{H})\right]-\mbb E \left[ \log\det\left(\mbf I_{N_{R}}+\frac{1}{\sigma^2}\tilde{\mbf G} \tilde{\mbf Q} \tilde{\mbf G}^H\right)\right]^2. 
\nonumber \end{align}
\underline{If $k<i:$}  Lemma \ref{meanproducti1i2} implies that
\begin{align}
        m(i,k) \leq c'\sqrt{\alpha}^{i-k}.
\nonumber \end{align} \underline{If $i<k:$} Lemma \ref{meanproducti1i2} implies that
\begin{align}
        m(i,k) \leq c'\sqrt{\alpha}^{k-i}.
\nonumber \end{align}

Therefore, we have

\begin{align}
    &\frac{1}{n^2}\sum_{i=1}^{n}\sum_{k=1}^{i-1} m(i,k) \nonumber \\
    &\leq \frac{c'}{n^2} \sum_{i=1}^{n}\sum_{k=1}^{i-1}\sqrt{\alpha}^{i-k} \nonumber \\ 
    &\leq \frac{c'}{n(1-\sqrt{\alpha})},\label{boundfirstpart}
\end{align}
because by Lemma \ref{sumterms1} in the Appendix, we have  for any $0<\alpha<1$
\begin{align}
    \sum_{i=1}^{n} \sum_{k=1}^{i-1} \alpha^{i-k}\leq \frac{n}{1-\alpha}.\nonumber
 \end{align}
Furthermore, it holds that 
\begin{align}
\frac{1}{n^2}\sum_{i=1}^{n}\sum_{k=i+1}^{n} m(i,k)&\leq \frac{c'}{n^2}\sum_{i=1}^{n}\sum_{k=i+1}^{n} \sqrt{\alpha}^{k-i} \nonumber \\  \nonumber \\
&\leq \frac{c'}{n(1-\sqrt{\alpha})} \label{boundsecondpart}
\end{align}
\color{black}
because by Lemma \ref{sumtermsiplus1} in the Appendix, we have for any $0<\alpha<1$
\begin{align}
    \sum_{i=1}^{n} \sum_{k=i+1}^{n} \alpha^{k-i}\leq\frac{n}{1-\alpha}. \nonumber
    \end{align}

From \eqref{boundfirstpart} and \eqref{boundsecondpart}, we deduce that

\begin{align} \frac{1}{n^2}\sum_{i=1}^{n}\sum_{k=1}^{i-1} m(i,k)+\frac{1}{n^2}\sum_{i=1}^{n}\sum_{k=i+1}^{n} m(i,k)
    &\leq \frac{2c'}{n(1-\sqrt{\alpha})}. \nonumber
\end{align}
\subsection{Upper-bound for $\frac{1}{n^2}\sum_{i=1}^{n}\mbb E\left[ i(\bs{T}_i;\bs{Z}_i,\mbf G_i)^2\right]$ for $0\leq \alpha <1$}
We are going to prove that
\begin{align}
    \frac{1}{n^2}\sum_{i=1}^{n}\mbb E\left[ i(\bs{T}_i;\bs{Z}_i,\mbf G_i)^2\right]\leq \frac{c''}{n} \nonumber
\end{align}
for some $c''>0.$
It suffices to show that $\mbb E\left[ i(\bs{T}_i;\bs{Z}_i,\mbf G_i)^2\right]$ is bounded from above for $i=1,\hdots,n.$ 
Recall that
\begin{align}
\bs{Z}_{i}=\mbf G_i\bs{T}_{i}+\bs{\xi}_{i}, \quad i=1\hdots n \nonumber
\end{align}
and that for $i=1\hdots n$
$$\bs{\xi}_{i} \sim \mathcal{N}_{\mathbb{C}}\left(\mbf 0_{N_{R}},\sigma^2\mbf I_{N_{R}}\right).$$
By Lemma \ref{infdensityi} in the Appendix, we know that for $i=1,\hdots, n$
\begin{align}
    &i(\bs{T}_i;\bs{Z}_i,\mbf G_i) \nonumber \\ 
    &=\log\det(\mbf I_{N_{R}}+\frac{1}{\sigma^2}\mbf G_i \tilde{\mbf Q} \mbf G_i^{H})-\frac{1}{\ln(2)\sigma^2}\left(\bs{Z}_i-\mbf G_i\bs{T}_i \right)^{H}\left(\bs{Z}_i-\mbf G_i \bs{T}_i\right)+\frac{1}{\ln(2)\sigma^2}\bs{Z}_i^{H} \left(\mbf I_{N_{R}}+\frac{1}{\sigma^2}\mbf G \tilde{\mbf Q} \mbf G^{H}  \right)^{-1}\bs{Z}_i. 
\nonumber \end{align}
We have
\begin{align} 
&\lvert i(\bs{T}_i;\bs{Z}_i,\mbf G_i)\rvert \nonumber \\ \nonumber \\
    &= \Bigg\lvert \log\det(\mbf I_{N_{R}}+\frac{1}{\sigma^2}\mbf G_i \tilde{\mbf Q} \mbf G_i^{H})-\frac{1}{\ln(2)\sigma^2}\left(\bs{Z}_i-\mbf G_i\bs{T}_i \right)^{H}\left(\bs{Z}_i-\mbf G_i \bs{T}_i\right)+\frac{1}{\ln(2)\sigma^2}\bs{Z}_i^{H} \left(\mbf I_{N_{R}}+\frac{1}{\sigma^2}\mbf G \tilde{\mbf Q} \mbf G^{H}  \right)^{-1}\bs{Z}_i \Bigg\rvert \nonumber \\ \nonumber \\
    &\leq \Bigg\vert\log\det(\mbf I_{N_{R}}+\frac{1}{\sigma^2}\mbf G_i \tilde{\mbf Q} \mbf G_i^{H})+\frac{1}{\ln(2)\sigma^2}\bs{Z}_i^{H} \left(\mbf I_{N_{R}}+\frac{1}{\sigma^2}\mbf G_i \tilde{\mbf Q} \mbf G_i^{H}  \right)^{-1}\bs{Z}_i \Bigg\rvert+\frac{1}{\ln(2)\sigma^2}\Bigg\lvert\left(\bs{Z}_i-\mbf G_i\bs{T}_i \right)^{H}\left(\bs{Z}_i-\mbf G_i \bs{T}_i\right)\Bigg\rvert \nonumber \\ \nonumber \\
    &=\Bigg\lvert\log\det(\mbf I_{N_{R}}+\frac{1}{\sigma^2}\mbf G_i \tilde{\mbf Q} \mbf G_i^{H})+\frac{1}{\ln(2)\sigma^2}\bs{Z}_i^{H} \left(\mbf I_{N_{R}}+\frac{1}{\sigma^2}\mbf G_i \tilde{\mbf Q} \mbf G_i^{H}  \right)^{-1}\bs{Z}_i \Bigg\rvert+\frac{1}{\ln(2)\sigma^2}\lvert\bs{\xi}_i^{H}\bs{\xi}_i\rvert \nonumber \\ \nonumber\\
    &=\Bigg\lvert\log\det(\mbf I_{N_{R}}+\frac{1}{\sigma^2}\mbf G_i \tilde{\mbf Q} \mbf G_i^{H})+\frac{1}{\ln(2)\sigma^2}\bs{Z}_i^{H} \left(\mbf I_{N_{R}}+\frac{1}{\sigma^2}\mbf G_i \tilde{\mbf Q} \mbf G_i^{H}  \right)^{-1}\bs{Z}_i \Bigg\rvert+\frac{1}{\ln(2)\sigma^2}\lVert\bs{\xi}_i\rVert^{2}.
\nonumber \end{align}
Since $i(\bs{T}_i;\bs{Z}_i,\mbf G_i) \in \mbb R, $ we have
 \begin{align}
    &i(\bs{T}_i;\bs{Z}_i,\mbf G_i)^{2} \nonumber \\ \nonumber \\
    &=\lvert i(\bs{T}_i;\bs{Z}_i,\mbf G_i)\rvert^{2} \nonumber \\ \nonumber \\
    &\leq \left(\Bigg\lvert\log\det(\mbf I_{N_{R}}+\frac{1}{\sigma^2}\mbf G_i \tilde{\mbf Q} \mbf G_i^{H})+\frac{1}{\ln(2)\sigma^2}\bs{Z}_i^{H} \left(\mbf I_{N_{R}}+\frac{1}{\sigma^2}\mbf G_i \tilde{\mbf Q} \mbf G_i^{H}  \right)^{-1}\bs{Z}_i \Bigg\rvert+\frac{1}{\ln(2)\sigma^2}\lVert\bs{\xi}_i\rVert^{2} \right)^2 \nonumber \\ \nonumber \\
    &\overset{(a)}{\leq} 2 \left(\Bigg\lvert\log\det(\mbf I_{N_{R}}+\frac{1}{\sigma^2}\mbf G_i \tilde{\mbf Q} \mbf G_i^{H})+\frac{1}{\ln(2)\sigma^2}\bs{Z}_i^{H} \left(\mbf I_{N_{R}}+\frac{1}{\sigma^2}\mbf G_i \tilde{\mbf Q} \mbf G_i^{H}  \right)^{-1}\bs{Z}_i \Bigg\rvert\right)^2+\frac{2}{\ln(2)^2\sigma^4}\lVert\bs{\xi}_i\rVert^{4} \nonumber \\ \nonumber \\
    &\overset{(b)}{\leq} 4\left[\log\det(\mbf I_{N_{R}}+\frac{1}{\sigma^2}\mbf G_i \tilde{\mbf Q} \mbf G_i^{H})\right]^2+\frac{4}{\ln(2)^2\sigma^4}\left(\bs{Z}_i^{H} (\mbf I_{N_{R}}+\frac{1}{\sigma^2}\mbf G_i \tilde{\mbf Q} \mbf G_i^{H} )^{-1}\bs{Z}_i\right)^2+\frac{2}{\ln(2)^2\sigma^4}\lVert\bs{\xi}_i\rVert^{4} \nonumber\\ \nonumber \\
    &\leq  4\left[\log\det(\mbf I_{N_{R}}+\frac{1}{\sigma^2}\mbf G_i \tilde{\mbf Q} \mbf G_i^{H})\right]^2+\frac{4}{\ln(2)^2\sigma^4}\lVert (\mbf I_{N_{R}}+\frac{1}{\sigma^2}\mbf G_i \tilde{\mbf Q} \mbf G_i^{H})^{-1} \rVert^2 \lVert \bs{Z}_i \rVert^{4} +\frac{2}{\ln(2)^2\sigma^4}\lVert\bs{\xi}_i\rVert^{4}\nonumber \\ \nonumber \\
    &\overset{(c)}{\leq}  4\left[\log\det(\mbf I_{N_{R}}+\frac{1}{\sigma^2}\mbf G_i \tilde{\mbf Q} \mbf G_i^{H})\right]^2+\frac{4}{\ln(2)^2\sigma^4} \lVert \bs{Z}_i \rVert^{4}+\frac{2}{\ln(2)^2\sigma^4}\lVert\bs{\xi}_i\rVert^{4} \nonumber \\ \nonumber \\
    &= 4\left[\log\det(\mbf I_{N_{R}}+\frac{1}{\sigma^2}\mbf G_i \tilde{\mbf Q} \mbf G_i^{H})\right]^2+\frac{4}{\ln(2)^2\sigma^4} \lVert \mbf G_i \bs{T}_i+\bs{\xi}_i\rVert^{4}+\frac{2}{\ln(2)^2\sigma^4}\lVert\bs{\xi}_i\rVert^{4} \nonumber \\ \nonumber \\
    &\leq4\left[\log\det(\mbf I_{N_{R}}+\frac{1}{\sigma^2}\mbf G_i \tilde{\mbf Q} \mbf G_i^{H})\right]^2+\frac{4}{\ln(2)^2\sigma^4} \left(\lVert \mbf G_i\rVert \lVert\bs{T}_i\rVert+\lVert\bs{\xi}_i\rVert\right)^{4}+\frac{2}{\ln(2)^2\sigma^4}\lVert\bs{\xi}_i\rVert^{4} \nonumber \\ \nonumber \\
    &\overset{(d)}{\leq} 4\left[\log\det(\mbf I_{N_{R}}+\frac{1}{\sigma^2}\lVert\mbf G_i \tilde{\mbf Q} \mbf G_i^{H}\rVert \mbf I_{N_{R}})\right]^2+\frac{4}{\ln(2)^2\sigma^4} \left(2 \lVert \mbf G_i \rVert^2\lVert \bs{T}_i \rVert^2+2\lVert \bs{\xi}_i \rVert^2  \right)^{2}+\frac{2}{\ln(2)^2\sigma^4}\lVert\bs{\xi}_i\rVert^{4} \nonumber \\ \nonumber \\
    &\overset{(e)}{\leq} \frac{4}{\ln(2)^2} \left[\mathrm{tr}\left( \frac{1}{\sigma^2}\lVert\mbf G_i \tilde{\mbf Q} \mbf G_i^{H}\rVert \mbf I_{N_{R}}\right)\right]^{2}+\frac{32}{\ln(2)^2\sigma^4} \left(\lVert \mbf G_i \rVert^4 \lVert \bs{T}_i \Vert^{4}+\lVert \bs{\xi}_{i}\rVert^4  \right)+\frac{2}{\ln(2)^2\sigma^4}\lVert\bs{\xi}_i\rVert^{4} \nonumber \\ \nonumber \\
    &\leq \frac{4}{\ln(2)^2\sigma^4}N_{R}^{2}\lVert \mbf G_i \rVert^{4}\lVert \tilde{\mbf Q} \rVert^2+\frac{32}{\ln(2)^2\sigma^4} \left(\lVert \mbf G_i \rVert^4 \lVert \bs{T}_i \Vert^{4}+\lVert \bs{\xi}_{i}\rVert^4  \right)+\frac{2}{\ln(2)^2\sigma^4}\lVert\bs{\xi}_i\rVert^{4} \nonumber \\ \nonumber \\
    &\overset{(f)}{\leq} \frac{4}{\ln(2)^2\sigma^4}N_{R}^{2}P^2\lVert \mbf G_i \rVert^{4}+\frac{32}{\ln(2)^2\sigma^4} \left(\lVert \mbf G_i \rVert^4 \lVert \bs{T}_i \Vert^{4}+\lVert \bs{\xi}_{i}\rVert^4  \right)+\frac{2}{\ln(2)^2\sigma^4}\lVert\bs{\xi}_i\rVert^{4},
\nonumber \end{align}
where $(a)(b)$ follow because for $K_{1},K_{2}\geq 0,$ $(K_{1}+K_{2})^2\leq 2K_{1}^2+2K_{2}^2,$ $(c)$ follows because $\lVert (\mbf I_{N_{R}}+\frac{1}{\sigma^2}\mbf G_i \tilde{\mbf Q} \mbf G_i^{H})^{-1} \rVert=\frac{1}{\lambda_{\min}(\mbf I_{N_{R}}+\frac{1}{\sigma^2}\mbf G_i \tilde{\mbf Q} \mbf G_i^{H})}\leq 1$, $(d)$ follows because $\mbf A \preceq \lVert \mbf A \rVert\mbf I_{n}$ for any Hermitian $\mbf A \in \mbb C^{n\times n}$ (by Lemma \ref{normposdefinite} in the Appendix),  $(e)$ follows because $\ln\det(\mbf I_{n}+\mbf A) \leq \mathrm{tr}(\mbf A)$ for $\mbf A$ positive semi-definite and because for $K_{1},K_{2}\geq 0,$ $(K_{1}+K_{2})^{2}\leq 2K_{1}^2+2K_{2}^2$ and $(f)$ follows because $\lVert \tilde{\mbf Q} \rVert=\lambda_{\max}(\tilde{\mbf Q})\leq \text{tr}(\tilde{\mbf Q})\leq P.$
This implies using the fact that $\mbf G_i$ and $\bs{T}_i$ are independent that
\begin{align}
  \mbb E \left[ i(\bs{T}_i;\bs{Z}_i,\mbf G_i)^2\right] 
  &\leq \frac{4P^2}{\ln(2)^2\sigma^4}N_{R}^{2}\mbb E \left[\lVert \mbf G_i \rVert^{4}\right]+\frac{32}{\ln(2)^2\sigma^4} \left(\mbb E \left[\lVert \mbf G_i \rVert^4\right] \mbb E\left[\lVert \bs{T}_i \Vert^{4}\right]+\mbb E\left[\lVert \bs{\xi}_{i}\rVert^4\right]  \right)+\frac{2}{\ln(2)^2\sigma^4}\mbb E\left[\lVert\bs{\xi}_i\rVert^{4}\right] \nonumber \\
  &\leq \frac{4P^2}{\ln(2)^2\sigma^4}N_{R}^{2}c_1+\frac{16}{\ln(2)^2\sigma^4} \left(c_1 c_2+c_3  \right)+\frac{2}{\ln(2)^2\sigma^4}c_3\nonumber \\
&=c'', \nonumber 
\end{align}
 for some $c_1,c_2,c_3>0,$ where we used that $\mbb E \left[\lVert \mbf G_i \rVert^4\right]<\infty$, $\mbb E\left[\lVert \bs{T}_i \Vert^{4}\right]<\infty$ and $\mbb E\left[\lVert \bs{\xi}_{i}\rVert^4\right]<\infty$ (by Lemma \ref{boundedmatrixnorm} in the Appendix) and where $c''>0.$
 
As a result, we have
\begin{align}
    \frac{1}{n^2}\sum_{i=1}^{n}\mbb E\left[ i(\bs{T}_i;\bs{Z}_i,\mbf G_i)^2\right]\leq \frac{c''}{n}.\nonumber
\end{align}
To summarize, we have proved that for $0<\alpha<1$
$$\frac{1}{n^2}\sum_{i=1}^{n}\sum_{k=1}^{i-1} m(i,k)+\frac{1}{n^2}\sum_{i=1}^{n}\sum_{k=i+1}^{n} m(i,k)
    \leq \frac{2c'}{n(1-\sqrt{\alpha})}$$
    and that for $0\leq \alpha<1$\\~\\
   $$\frac{1}{n^2}\sum_{i=1}^{n}\mbb E\left[ i(\bs{T}_i;\bs{Z}_i,\mbf G_i)^2\right]\leq \frac{c''}{n}.$$

Now, from \eqref{sumterms}, we know that
\begin{align}
    &\mathrm{var}\left(\frac{i(\bs{T}^n;\bs{Z}^n,\mbf G^n)}{n} \right)\nonumber \\ &\leq 
        \begin{cases}
    \frac{1}{n^2}\sum_{i=1}^{n}\mbb E\left[ i(\bs{T}_i;\bs{Z}_i,\mbf G_i)^2\right], & \text{for } \alpha=0 \\~\\
        \frac{1}{n^2}\sum_{i=1}^{n}\sum_{k=1}^{i-1} m(i,k)+\frac{1}{n^2}\sum_{i=1}^{n}\sum_{k=i+1}^{n} m(i,k)+\frac{1}{n^2}\sum_{i=1}^{n}\mbb E\left[ i(\bs{T}_i;\bs{Z}_i,\mbf G_i)^2\right] & \text{for } 0<\alpha<1.
\end{cases}
\nonumber
\end{align}

To conclude, it follows that
\begin{align}
    \mathrm{var}\left(\frac{i(\bs{T}^n;\bs{Z}^n,\mbf G^n)}{n} \right)&\leq \begin{cases}
      \frac{c''}{n}, & \text{for } \alpha=0 \\
        \frac{2c'}{n(1-\sqrt{\alpha})}+ \frac{c''}{n}, & \text{for } 0<\alpha<1
\end{cases}
\nonumber \\
    &=\kappa(n,\alpha),
\nonumber \end{align}
where $\underset{n\rightarrow \infty}{\lim} \kappa(n,\alpha)=0.$  This completes the proof of Lemma \ref{boundvariance}.
\section{Conclusion}
\label{conclusion}
In this paper, we studied the problem of message transmission over time-varying MIMO first-order Gauss-Markov Rayleigh fading channels with average power constraint and with CSIR, as an example of infinite-state Markov fading channels. The Gauss-Markov model is widely used to describe the time-varying aspect of the channel. The novelty of our work lies in establishing a single-letter characterization of the channel capacity. The capacity formula that we proved coincides with the one of i.i.d. MIMO Rayleigh fading channels. Therefore, the memory factor has no influence on the capacity. The CSIR assumption is critical in deriving the single-letter expression for the capacity. In the absence of that assumption, no single-letter formula for the capacity exists even for SISO finite-state Markov fading channels.
As a future work, it would be interesting to study the capacity of time-varying MIMO Rayleigh fading channels when a higher-order Gauss-Markov model is used to describe the channel variations over the time.
\section{ACKNOWLEDGEMENTS}
H. Boche was supported in part by the German Federal Ministry of Education and Research (BMBF) within the national initiative on 6G Communication Systems through the research hub 6G-life under Grant 16KISK002, within the national initiative on Post Shannon Communication
(NewCom) under Grant 16KIS1003K. He was further supported by the German Research Foundation (DFG) within Germany’s Excellence Strategy EXC-2092–390781972.

M. Wiese was supported by the Deutsche Forschungsgemeinschaft (DFG, German
Research Foundation) within the Gottfried Wilhelm Leibniz Prize under Grant BO 1734/20-1, and within Germany’s Excellence Strategy EXC-2111-390814868 and EXC-2092 CASA-390781972.

C.\ Deppe was supported in part by the German Federal Ministry of Education and Research (BMBF) under Grant 16KIS1005 and in part by the German Federal Ministry of Education and Research (BMBF) within the national initiative on 6G Communication Systems through the research hub 6G-life under Grant 16KISK002.

R. Ezzine was supported by the German Federal Ministry of Education and Research (BMBF) under Grant 16KIS1003K.
\appendix
\subsection{Auxiliary Lemmas}
\begin{lemma}
\label{normposdefinite}
For any Hermitian matrix $\mbf A \in \mbb C^{n\times n}$, the matrix
$$\lVert \mbf A \rVert \mbf I_{n} - \mbf A $$
is positive semi-definite.
\end{lemma}
\color{black}
\begin{proof}
Since $\mbf A$ is Hermitian, we know that for any $\bs{x}\in \mbb C^{n},$ $\bs{x}^{H}\mbf A \bs{x}$ is real.
Therefore, for any $\bs{x}\in \mbb C^{n} \setminus \{\bs{0}\},$
\begin{align}
 \bs{x}^{H}\mbf A\bs{x} &\leq  \lvert \bs{x}^{H}\mbf A\bs{x}\rvert \nonumber \\  \nonumber \\
&\leq \lVert \mbf A \rVert \lVert  \bs{x} \rVert^{2}. \nonumber
\end{align}
It follows that
\begin{align}
\bs{x}^{H}\left( \lVert \mbf A \rVert \mbf I_{n}-\mbf A \right) \bs{x} 
&=\bs{x}^{H}\lVert\mbf A\rVert\mbf I_{n}\bs{x}-\bs{x}^{H}\mbf A\bs{x} \nonumber \\  \nonumber \\
&=\lVert \mbf A \rVert \lVert  \bs{x} \rVert^{2}-\bs{x}^{H}\mbf A\bs{x}\nonumber\\  \nonumber \\
&\geq 0.
\nonumber \end{align}
\end{proof}
\color{black}
\begin{lemma}
\label{bounddeterminantsum}
Let $\mbf A \in \mbb C^{n\times n}$ be any positive-definite Hermitian matrix with $\lambda_{\min}(\mbf A)$ being its smallest eigenvalue and let $\mbf B \in \mbb C^{n \times n}$ be any positive semi-definite matrix, then
\begin{align}
    \log\det(\mbf A+\mbf B)\leq\log\det(\mbf A)+\log\det(\mbf I_{n}+\frac{1}{\lambda_{\min}(\mbf A)}\mbf B).
\nonumber \end{align}
\end{lemma}
\color{black}
\begin{proof}
\begin{align}
\det(\mbf A+\mbf B)&=\det(\mbf A)\det(\mbf I_{n}+\mbf A^{-1}\mbf B) \nonumber \\
&=\det(\mbf A)\det(\mbf I_{n}+\mbf B\mbf A^{-1}) \nonumber\\
&=\det(\mbf A)\det(\mbf I_{n}+\mbf B^{\frac{1}{2}}\mbf B^{\frac{1}{2}}\mbf A^{-1}) \nonumber\\
&= \det(\mbf A)\det(\mbf I_{n}+\mbf B^{\frac{1}{2}}\mbf A^{-1}\mbf B^{\frac{1}{2}}) \nonumber \\
&\overset{(a)}{\leq} \det(\mbf A)\det(\mbf I_{n}+\mbf B^{\frac{1}{2}}\frac{1}{\lambda_{\min}(\mbf A)}\mbf I_{n}\mbf B^{\frac{1}{2}}) \nonumber \\
&= \det(\mbf A)\det(\mbf I_{n}+ \frac{1}{\lambda_{\min}(\mbf A)}\mbf I_{n}\mbf B) \nonumber \\
&=\det(\mbf A)\det(\mbf I_{n}+\frac{1}{\lambda_{\min}(\mbf A)}\mbf B),
\nonumber \end{align}
\color{black}
where $(a)$ follows from the following properties\cite{normposdef}:
\begin{enumerate}
\item For any positive semi-definite Hermitian matrices $\mbf M_1$ and $\mbf M_2,$ if $\mbf M_1- \mbf M_2$ is Hermitian positive semi-definite then
$$\det\left(\mbf M_1\right)\geq \det(\mbf M_2). $$
\item For any Hermitian matrix $\mbf M \in \mbb C^{n\times n},$ with minimum eigenvalue $\lambda_{\min}(\mbf M),$ it holds that
    $$\mbf M-\lambda_{\min}(\mbf M)\mbf I$$ is positive semi-definite, 

\item For any positive definite Hermitian matrices $\mbf M$ and $\tilde{\mbf M},$  $\mbf M- \tilde{\mbf M}$ is positive semi-definite if and only if
$\tilde{\mbf M}^{-1}-\mbf M^{-1}$  is positive semi-definite.
\end{enumerate}
\color{black}
Therefore, it follows that
\begin{align}
      \log\det(\mbf A+\mbf B)\leq\log\det(\mbf A)+\log\det(\mbf I_{n}+\frac{1}{\lambda_{\min}(\mbf A)}\mbf B).\nonumber
\end{align}
\end{proof}
\color{black}

\color{black}
\begin{lemma}
$\mbb E \left[\Lambda\left( \mbf G_{i_{1}},\mbf S,\tilde{\mbf G},\tilde{\mbf W}\right)\right]\leq c $ for some $c>0.$
\label{meanbounded}
\end{lemma}
\begin{proof}
We recall that 
\begin{align}
    \Lambda\left( \mbf G_{i_{1}},\mbf S,\tilde{\mbf G},\tilde{\mbf W}\right) =\lVert \mbf G_{i_{1}} \rVert^{2} \left(P\lVert \mbf G_{i_{1}} \rVert^2+P\lVert\tilde{\mbf G}\rVert^2+2\lVert \tilde{\mbf W}\tilde{\mbf Q}\tilde{\mbf G}^H \rVert +2P\lVert\mbf G_{i_{1}}\rVert \lVert\mbf S\rVert          \right), \nonumber
\end{align}
where $$\tilde{\mbf W}=\mbf S+\sqrt{\alpha}^{i_{2}-i_{1}} \tilde{\mbf G}$$
with $i_1<i_2,$ where
 $\mbf S=\sqrt{1-\alpha}\sum_{j=i_{1}+1}^{i_{2}} \sqrt{\alpha}^{i_{2}-j} \mbf W_j,$
and where $\tilde{\mbf G}$ is a random matrix  independent of $\mbf G_1$ and $\mbf W_i, \ i=2,\hdots n$ such that $\vect(\tilde{\mbf G})\sim \mc N_{\mbb C}\left(\bs{0}_{N_RN_T},\mbf K \right).$

We have
\begin{align}
    &\mbb E \left[\Lambda\left( \mbf G_{i_{1}},\mbf S,\tilde{\mbf G},\tilde{\mbf W}\right)\right] \nonumber\\ \nonumber \\
    &=\mbb E \left[\lVert \mbf G_{i_{1}} \rVert^{2} \left(P\lVert \mbf G_{i_{1}} \rVert^2+P\lVert\tilde{\mbf G}\rVert^2+2\lVert \tilde{\mbf W}\tilde{\mbf Q}\tilde{\mbf G}^H \rVert +2P\lVert\mbf G_{i_{1}}\rVert \lVert\mbf S\rVert          \right)\right] \nonumber \\ \nonumber \\
    &= P\mbb E \left[ \lVert \mbf G_{i_{1}} \rVert^{4}\right]+P\mbb E\left[\lVert \mbf G_{i_{1}} \rVert^{2}\lVert \tilde{\mbf G} \rVert^2 \right]+2\mbb E \left[\lVert \mbf G_{i_{1}} \rVert^{2}\lVert \tilde{\mbf W}\tilde{\mbf Q}\tilde{\mbf G}^H \rVert \right]+2P\mbb E \left[\lVert \mbf G_{i_{1}} \rVert^{3}\lVert \mbf S \rVert \right] \nonumber \\ \nonumber \\
    &= P\mbb E \left[ \lVert \mbf G_{i_{1}} \rVert^{4}\right]+P\mbb E\left[\lVert \mbf G_{i_{1}} \rVert^{2}\right]\mbb E \left[\lVert \tilde{\mbf G} \rVert^2 \right]+2\mbb E \left[\lVert \mbf G_{i_{1}} \rVert^{2}\right]\mbb E\left[\lVert \tilde{\mbf W}\tilde{\mbf Q}\tilde{\mbf G}^H \rVert \right]+2P\mbb E \left[\lVert \mbf G_{i_{1}} \rVert^{3}\right]\mbb E\left[\lVert \mbf S \rVert \right] \nonumber \\ \nonumber\\
    &= P\mbb E \left[ \lVert \tilde{\mbf G} \rVert^{4}\right]+P\mbb E \left[\lVert \tilde{\mbf G} \rVert^2 \right]^2+2\mbb E \left[\lVert \tilde{\mbf G} \rVert^{2}\right]\mbb E\left[\lVert \tilde{\mbf W}\tilde{\mbf Q}\tilde{\mbf G}^H \rVert \right]+2P\mbb E \left[\lVert \tilde{\mbf G} \rVert^{3}\right]\mbb E\left[\lVert \mbf S \rVert \right],
    \nonumber
\end{align}
where we used that $\mbf G_{i_{1}}$ is independent of $(\tilde{\mbf W},\tilde{\mbf G})$ and that $\tilde{\mbf G}$ has the same distribution as $\mbf G_{i_{1}}.$

By Lemma \ref{boundedmatrixnorm}, we know that $\mbb E\left[ \lVert \tilde{\mbf G} \rVert^{\ell}\right]<\infty$ for $\ell\in\{2,3,4\}.$ Therefore, to complete the proof, we have to show that $\mbb E\left[\lVert \tilde{\mbf W}\tilde{\mbf Q}\tilde{\mbf G}^H \rVert \right]$ and $\mbb E\left[\lVert \mbf S\rVert\right]$ are both bounded from above. 

It holds that
\begin{align}
 &\mbb E\left[\lVert \tilde{\mbf W}\tilde{\mbf Q}\tilde{\mbf G}^H \rVert \right] \nonumber \\ \nonumber \\
 &=\mbb E\left[\Bigg\lVert \left(\mbf S+\sqrt{\alpha}^{i_{2}-i_{1}} \tilde{\mbf G}\right)\tilde{\mbf Q}\tilde{\mbf G}^H \Bigg\rVert \right] \nonumber \\ \nonumber \\
 &=\mbb E\left[\Bigg\lVert \mbf S\tilde{\mbf Q}\tilde{\mbf G}^H+\sqrt{\alpha}^{i_{2}-i_{1}} \tilde{\mbf G}\tilde{\mbf Q}\tilde{\mbf G}^H \Bigg\rVert \right] \nonumber \\ \nonumber \\
 &\leq \mbb E\left[\lVert \mbf S\tilde{\mbf Q}\tilde{\mbf G}^H\rVert+\sqrt{\alpha}^{i_{2}-i_{1}} \lVert\tilde{\mbf G}\tilde{\mbf Q}\tilde{\mbf G}^H \rVert \right] \nonumber \\ \nonumber \\
 &\leq \lVert \tilde{\mbf Q}\rVert \mbb E\left[\lVert \mbf S\rVert \lVert\tilde{\mbf G}\rVert+ \lVert\tilde{\mbf G}\rVert^2 \right] \nonumber \\ \nonumber \\ 
 &\leq P \mbb E\left[\lVert \mbf S\rVert \lVert\tilde{\mbf G}\rVert+ \lVert\tilde{\mbf G}\rVert^2 \right] \nonumber \\ \nonumber \\
 &=P\left( \mbb E\left[\lVert \mbf S\rVert\right]\mbb E\left[\lVert \tilde{\mbf G}\rVert\right]+\mbb E\left[\lVert \tilde{\mbf G}\rVert^2 \right]     \right), \nonumber
\end{align}
where we used that $\tilde{\mbf G}$ and $\mbf S$ are independent in the last step, since $\tilde{\mbf G}$ and $\mbf W_{i_{1}+1},\hdots \mbf W_{i_{2}}$ are independent.

Now, from Lemma \ref{boundedmatrixnorm}, we know that $\mbb E\left[\lVert \tilde{\mbf G}\rVert \right]<\infty.$ Therefore, to complete the proof, it suffices to show that $\mbb E\left[ \lVert \mbf S \rVert \right]$ is bounded from above..

We have
\begin{align}
&\mbb E \left[ \lVert \mbf S \rVert \right] \nonumber \\
&=\mbb E \left[ \Bigg\lVert \sqrt{1-\alpha}\sum_{j=i_{1}+1}^{i_{2}} \sqrt{\alpha}^{i_{2}-j} \mbf W_j \Bigg\rVert \right] \nonumber \\
&\leq \mbb E\left[ \sqrt{1-\alpha}\sum_{j=i_{1}+1}^{i_{2}}\sqrt{\alpha}^{i_{2}-j}\lVert \mbf W_{j}\rVert  \right]\nonumber \\
&=\sqrt{1-\alpha}\sum_{j=i_{1}+1}^{i_{2}}\sqrt{\alpha}^{i_{2}-j}\mbb E\left[\lVert \mbf W_{j}\rVert\right] \nonumber \\
&=\sqrt{1-\alpha}\mbb E\left[\lVert \tilde{\mbf G} \rVert\right]\sum_{j=i_{1}+1}^{i_{2}}\sqrt{\alpha}^{i_{2}-j} \nonumber \\
&=\frac{\sqrt{1-\alpha}}{1-\sqrt{\alpha}}\left(1-\sqrt{\alpha}^{i_{2}-i_{1}}\right)\mbb E\left[\lVert \tilde{\mbf G} \rVert\right] \nonumber \\
&\leq \frac{\sqrt{1-\alpha}}{1-\sqrt{\alpha}}\mbb E\left[\lVert \tilde{\mbf G} \rVert\right], \nonumber
\end{align}
where we used that
\begin{align}
\sum_{j=i_{1}+1}^{i_{2}}\sqrt{\alpha}^{i_{2}-j}&=\sqrt{\alpha}^{i_2}\sum_{j=i_{1}+1}^{i_{2}}\left(\frac{1}{\sqrt{\alpha}}\right)^{j} \nonumber \\
&=\sqrt{\alpha}^{i_{2}}\left(\frac{1}{\sqrt{\alpha}}\right)^{i_{1}+1}\frac{1-\left(\frac{1}{\sqrt{\alpha}}\right)^{i_{2}-i_{1}}}{1-\frac{1}{\sqrt{\alpha}}} \nonumber \\
&= \frac{\sqrt{\alpha}^{i_{2}-i_{1}}-1}{\sqrt{\alpha}-1}
\nonumber \\
&=\frac{1-\sqrt{\alpha}^{i_{2}-i_{1}}}{1-\sqrt{\alpha}}
\nonumber 
\end{align}
and that $\tilde{\mbf G}$ has the same distribution as each of the $\mbf W_i.$ 
Therefore, $\mbb E \left[ \lVert \mbf S \rVert \right]$ is bounded from above.
This proves that $\mbb E \left[\Lambda\left( \mbf G_{i_{1}},\mbf S,\tilde{\mbf G},\tilde{\mbf W}\right)\right]\leq c$ for some $c>0.$
\end{proof}

\color{black} 
\color{black}
 \color{black}
\begin{lemma}
\label{boundedmatrixnorm}
For any $\mbf K \in \mbb C^{N_RN_T\times N_RN_T},$  
let $\mbf G \in \mbb C^{N_{R} \times N_{T}}$
a random matrix such that 
\begin{align}
    \vect\left(\mbf G\right)\sim \mc N_{\mbb C}\left(\bs{0}_{N_{R}N_{T}},\mbf K\right).
\nonumber \end{align}
Then for $\ell\in \{1,2,3,4\},$ it holds that  $\mbb E \left[ \lVert \mbf G \rVert^{\ell}\right]<\infty.$ 
\end{lemma}
\begin{proof}
Let $G^{(1)},\hdots,G^{(N_{R}N_{T})}$ be the entries of $\text{vec}(\mbf G).$

For $\ell=2,$ we have
\begin{align}
    \lVert \mbf G \rVert^2&=\lambda_{\max}(\mbf G \mbf G^H) \nonumber \\
    &\leq \text{tr}(\mbf G \mbf G^H) \nonumber \\
    &=\sum_{i=1}^{N_{R}N_{T}} \lvert  G^{(i)} \rvert^{2}.\label{upperboundnormsquare}
\end{align}
Therefore
\begin{align}
    \mbb E \left[ \lVert \mbf G \rVert^2 \right]&\leq \sum_{i=1}^{N_{R}N_{T}} \mbb E \left[\lvert  G^{(i)} \rvert^{2}\right] \nonumber \\
   &<\infty,\label{secondmomentbounded}
\end{align}
where we used that the second moment of each $G^{(i)}, i=1,\hdots,N_RN_T,$ is finite.

For $\ell=4,$ it follows using \eqref{upperboundnormsquare} that
\begin{align}
\lVert \mbf G \rVert^4&\leq \left( \sum_{i=1}^{N_{R}N_{T}} \lvert  G^{(i)} \rvert^{2}\right)^2 \nonumber \\
&=\sum_{i=1}^{N_R N_T}\sum_{j=1,j\neq i}^{N_R N_T} \lvert G^{(i)} \rvert^{2}\lvert G^{(j)} \rvert^{2} +\sum_{i=1}^{N_R N_T} \lvert G^{(i)}\rvert^4. \nonumber
\end{align}
This yields
\begin{align}
   \mbb E \left[\lVert \mbf G\rVert^{4} \right] &\leq \sum_{i=1}^{N_R N_T}\sum_{j=1,j\neq i}^{N_R N_T} \mbb E \left[\lvert G^{(i)} \rvert^{2}\lvert G^{(j)} \rvert^{2}\right] +\sum_{i=1}^{N_R N_T} \mbb E \left[ \lvert G^{(i)}\rvert^4 \right] \nonumber \\
 &\overset{(a)}{\leq} 
 \sum_{i=1}^{N_R N_T}\sum_{j=1,j\neq i}^{N_R N_T} \sqrt{\mbb E \left[\lvert G^{(i)} \rvert^{4}\right]\mbb E\left[\lvert G^{(j)} \rvert^{4}\right]} +\sum_{i=1}^{N_R N_T} \mbb E \left[ \lvert G^{(i)}\rvert^4 \right] ´\nonumber \\
 &\overset{(b)}{<}\infty,\label{boundedfourthermoment}
\end{align}
where $(a)$ follows from Cauchy Schwarz's inequality and $(b)$ follows because the fourth moment of each $G^{(i)}, i=1,\hdots,N_RN_T,$ is finite.

For $\ell=1,$ we know that
\begin{align*}
    \mbb E\left[\lVert \mbf G\rVert \right]&\leq \sqrt{\mbb E\left[\lVert \mbf G\rVert^2 \right]} \nonumber \\
    &<\infty,
\end{align*}
where we used \eqref{secondmomentbounded} in the last step.

For $\ell=3,$ it holds that
\begin{align*}
    \mbb E\left[\lVert \mbf G\rVert^3 \right] &=\mbb E\left[\lVert \mbf G\rVert\lVert \mbf G\rVert^2\right] \nonumber \\
    &\overset{(a)}{\leq} \sqrt{\mbb E\left[\lVert\mbf G\rVert^2 \right]\mbb E\left[\lVert \mbf G \rVert^4 \right]} \nonumber \\
    &\overset{(b)}{<}\infty,
\end{align*}
where $(a)$ follows from Cauchy Schwarz's inequality and $(b)$ follows from \eqref{secondmomentbounded} and \eqref{boundedfourthermoment}.
\end{proof}
\color{black}
\begin{lemma}
\label{sumterms1}
For any $0<\alpha<1$,  it holds that
\begin{align}
    \sum_{i=1}^{n} \sum_{k=1}^{i-1} \alpha^{i-k}\leq \frac{n}{1-\alpha}.
\nonumber \end{align}
\end{lemma}
\color{black}
\begin{proof}
We have
\begin{align}
    &\sum_{i=1}^{n} \sum_{k=1}^{i-1} \alpha^{i-k} \nonumber \\
    &=\sum_{i=1}^{n} \alpha^{i} \sum_{k=1}^{i-1}\left(\frac{1}{\alpha}\right)^{k} \nonumber \\
    &=\sum_{i=1}^{n}\alpha^{i}\frac{1}{\alpha}\frac{1-\left(\frac{1}{\alpha}\right)^{i-1}}{1-\frac{1}{\alpha}} \nonumber\\
    &=\sum_{i=1}^{n}\alpha^{i}\frac{1-\left(\frac{1}{\alpha}\right)^{i-1}}{\alpha-1} \nonumber \\
    &=\sum_{i=1}^{n}\frac{\alpha^i-\alpha}{\alpha-1} \nonumber \\
   &=\sum_{i=1}^{n}\frac{\alpha-\alpha^i}{1-\alpha} \nonumber \\
   &=\frac{n\alpha}{1-\alpha}-\sum_{i=1}^{n}\frac{\alpha^i}{1-\alpha} \nonumber \\
   &\leq \frac{n\alpha}{1-\alpha} \nonumber \\
   &\leq \frac{n}{1-\alpha}. \nonumber
\end{align}
\end{proof}
\begin{lemma}
\label{sumtermsiplus1}
For any $0<\alpha<1$ it holds that
\begin{align}
    \sum_{i=1}^{n} \sum_{k=i+1}^{n} \alpha^{k-i}\leq\frac{n}{1-\alpha}.
\nonumber \end{align}
\end{lemma}
\color{black}
\begin{proof}
We have
\begin{align}
    &\sum_{i=1}^{n} \sum_{k=i+1}^{n} \alpha^{k-i} \nonumber \\
    &=\sum_{i=1}^{n} \left(\frac{1}{\alpha}\right)^{i} \sum_{k=i+1}^{n} \alpha^{k} \nonumber \\
    &=\sum_{i=1}^{n} \left(\frac{1}{\alpha}\right)^{i} \alpha^{i+1} \frac{\left(1-\alpha^{n-i}\right)}{1-\alpha} \nonumber\\
    &=\sum_{i=1}^{n}\frac{\alpha\left(1-\alpha^{n-i}\right)}{1-\alpha} \nonumber \\
    &=\frac{n\alpha}{1-\alpha}-\frac{\alpha^{n+1}}{1-\alpha}\sum_{i=1}^{n}\left(\frac{1}{\alpha}\right)^{i} \nonumber \\
     &\leq \frac{n\alpha}{1-\alpha} \nonumber \\
   &\leq \frac{n}{1-\alpha}. \nonumber
\nonumber \end{align}
\end{proof}
\color{black}
\begin{lemma} 
\label{infdensityi}
$\forall i \in \{1,\hdots,n\}$
\begin{align}
&i(\bs{T}_i;\bs{Z}_i,\mbf G_i) \nonumber \\  \nonumber \\
    &=\log\det(\mbf I_{N_{R}}+\frac{1}{\sigma^2}\mbf G_i \tilde{\mbf Q} \mbf G_i^{H})-\frac{1}{\ln(2)\sigma^2}\left(\bs{Z}_i-\mbf G_i\bs{T}_i \right)^{H}\left(\bs{Z}_i-\mbf G_i \bs{T}_i\right)+\frac{1}{\ln(2)\sigma^2}\bs{Z}_i^{H} \left(\mbf I_{N_{R}}+\frac{1}{\sigma^2}\mbf G_i \tilde{\mbf Q} \mbf G_i^{H}  \right)^{-1}\bs{Z}_i, \nonumber
    \end{align}
    where $\bs{T}_i\sim \mc N_{\mbb C}\left(\bs{0}_{N_T},\tilde{\mbf Q}\right), i=1\hdots n.$
\end{lemma}
\color{black}
\begin{proof}
Notice that
\begin{align}
 i(\bs{T}_i;\bs{Z}_i,\mbf G_i) &=\log\left(  \frac{ p_{\bs{Z}_i,\mbf G_i,\bs{T}_i}\left(\bs{Z}_i,\mbf G_i ,\bs{T}_i\right)   }{ p_{\bs{Z}_i,\mbf G_i}\left( \bs{Z}_i,\mbf G_i\right)p_{\bs{T}_i}(\bs{T}_i) }   \right) \nonumber \\  \nonumber \\
 &=\log\left( \frac{ p_{\bs{Z}_i|\mbf G_i,\bs{T}_i}\left(\bs{Z}_i|\mbf G_i ,\bs{T}_i\right)   }{ p_{\bs{Z}_i|\mbf G_i}\left( \bs{Z}_i|\mbf G_i\right)}   \right),\nonumber 
\end{align}
where we used that $\bs{T}_i$ and $\mbf G_i$ are independent.
 
It holds that
\begin{align}
    \bs{Z}_i|\mbf G_i,\bs T_i \sim \mc N_{\mbb C}\left(\mbf G_i \bs{T}_i,\sigma^2\mbf I_{N_R}\right)
\nonumber \end{align}
and that
\begin{align}
    \bs{Z}_i|\mbf G_i \sim \mc N_{\mbb C} \left( \bs{0}_{N_{R}},\mbf G_i \tilde{\mbf Q} \mbf G_i^{H} + \sigma^2 \mbf I_{N_{R}}        \right).
\nonumber \end{align}
It follows that
\begin{align}
&\log \frac {p_{\bs{Z}_i|\mbf G_i,\bs{T}_i}(\bs{Z}_i|\mbf G_i,\bs{T}_i)}{p_{\bs{Z}_i|\mbf G_i}(\bs{Z}_i|\mbf G_i)} \nonumber\\  \nonumber \\
&=\log\left[\frac{\frac{1}{\pi^{N_R}\det(\sigma^2\mbf I_{N_{R}})}\exp\left( \frac{-1}{\sigma^2}\left(\bs{Z}_i-\mbf G_i\bs{T}_i \right)^{H}\left(\bs{Z}_i-\mbf G_i \bs{T}_i\right)  \right)}{\frac{1}{\pi^{N_R}\det(\mbf G_i \tilde{\mbf Q} \mbf G_i^{H} + \sigma^2 \mbf I_{N_{R}})} \exp\left(-\frac{1}{\sigma^2}\bs{Z}_i^{H} \left(\mbf I_{N_{R}}+\frac{1}{\sigma^2}\mbf G_i \tilde{\mbf Q} \mbf G_i^{H}  \right)^{-1}\bs{Z}_i  \right)}\right] \nonumber \\  \nonumber \\
&=\log \left[\frac{\det(\mbf G_i \tilde{\mbf Q} \mbf G_i^{H} + \sigma^2 \mbf I_{N_{R}})}{\det(\sigma^2\mbf I_{N_{R}})}2^{\left( \frac{-1}{\ln(2)\sigma^2}\left(\bs{Z}_i-\mbf G_i\bs{T}_i \right)^{H}\left(\bs{Z}_i-\mbf G_i \bs{T}_i\right) +\frac{1}{\ln(2)\sigma^2}\bs{Z}_i^{H} \left(\mbf I_{N_{R}}+\frac{1}{\sigma^2}\mbf G_i \tilde{\mbf Q} \mbf G_i^{H}  \right)^{-1}\bs{Z}_i  \right)}\right] \nonumber \\  \nonumber \\
&=\log\det(\mbf I_{N_{R}}+\frac{1}{\sigma^2}\mbf G_i \tilde{\mbf Q} \mbf G_i^{H})-\frac{1}{\ln(2)\sigma^2}\left(\bs{Z}_i-\mbf G_i\bs{T}_i \right)^{H}\left(\bs{Z}_i-\mbf G_i \bs{T}_i\right)+\frac{1}{\ln(2)\sigma^2}\bs{Z}_i^{H} \left(\mbf I_{N_{R}}+\frac{1}{\sigma^2}\mbf G_i \tilde{\mbf Q} \mbf G_i^{H}  \right)^{-1}\bs{Z}_i.
\nonumber \end{align}
\end{proof}
\color{black}

\vspace{12pt}

\begin{thebibliography}{00}
\bibitem{tactiles}G. P. Fettweis, "The Tactile Internet: Applications and Challenges," in IEEE Vehicular Technology Magazine, vol. 9, no. 1, pp. 64-70, March 2014.
\bibitem{industrie40}Y. Lu, “Industry 4.0: A survey on technologies, applications and open
research issues,” Journal of Industrial Information Integration, vol. 6,
pp. 1–10, 2017.
\bibitem{Gallager}R. G. Gallager, Information Theory and Reliable Communication. New
York, NY, USA: John Wiley \& Sons, Inc., 1968.
\bibitem{survey}P. Sadeghi, R. A. Kennedy, P. B. Rapajic and R. Shams, "Finite-state Markov modeling of fading channels - a survey of principles and applications," in IEEE Signal Processing Magazine, vol. 25, no. 5, pp. 57-80, September 2008.
\bibitem{Goldsmith}A. J. Goldsmith and P. P. Varaiya, "Capacity, mutual information, and coding for finite-state Markov channels," in IEEE Transactions on Information Theory, vol. 42, no. 3, pp. 868-886, May 1996.
\bibitem{capacityanalysis} P. Sadeghi and P. Rapajic, "Capacity analysis for finite-state Markov mapping of flat-fading channels," in IEEE Transactions on Communications, vol. 53, no. 5, pp. 833-840, May 2005.
\bibitem{communicationsreceiverbased}M. Riediger and E. Shwedyk, "Communication receivers based on Markov models of the fading channel," IEEE CCECE2002. Canadian Conference on Electrical and Computer Engineering. Conference Proceedings (Cat. No.02CH37373), vol. 3, pp. 1255-1260, 2002.
\bibitem{firstordermarkov}M. R. Hueda, "On first-order Markov modeling for block errors on fading channels," Vehicular Technology Conference. IEEE 55th Vehicular Technology Conference. VTC Spring 2002 (Cat. No.02CH37367), 2002, pp. 1336-1339
\bibitem{finitestatewithfeedback}J. Chen and T. Berger, "The capacity of finite-State Markov Channels With feedback," in IEEE Transactions on Information Theory, vol. 51, no. 3, pp. 780-798, March 2005.
\bibitem{finitestaterayleighfading}Q. Zhang and S. A. Kassam, "Finite-state Markov model for Rayleigh fading channels," in IEEE Transactions on Communications, vol. 47, no. 11, pp. 1688-1692, Nov. 1999.
\bibitem{markovdelayedfeedback}H. Viswanathan, "Capacity of Markov channels with receiver CSI and delayed feedback," in IEEE Transactions on Information Theory, vol. 45, no. 2, pp. 761-771, March 1999.
\bibitem{mimorayleighfading}Y. Zhao, M. Zhao, L. Xiao and J. Wang, "Capacity of Time-Varying Rayleigh Fading MIMO Channels," 2005 IEEE 16th International Symposium on Personal, Indoor and Mobile Radio Communications, pp. 547-551, 2005.
\bibitem{gaussmarkov1}S. Misra, A. Swami and L. Tong, "Optimal training over the Gauss-Markov fading channel: a cutoff rate analysis," 2004 IEEE International Conference on Acoustics, Speech, and Signal Processing, pp. iii-809, 2004.
\bibitem{gaussmarkov2}C. B. Peel and A. L. Swindlehurst, "Capacity-optimal training for space-time modulation over a time-varying channel," IEEE International Conference on Communications, 2003. ICC '03., vol.5., pp. 3036-3040, 2003.
\bibitem{gaussmarkov3}M. Medard, "The effect upon channel capacity in wireless communications of perfect and imperfect knowledge of the channel," in IEEE Transactions on Information Theory, vol. 46, no. 3, pp. 933-946, May 2000.
\bibitem{mimobenefits}A. Katalinic, R. Nagy, and R. Zentner, “Benefits of MIMO systems
in practice: Increased capacity, reliability and spectrum efficiency,” in
Proceedings ELMAR 2006, pp. 263–266, 2006.
\bibitem{telatar} E. Telatar, “Capacity of Multi-antenna Gaussian Channels,” European Transactions on Telecommunications, vol. 10, pp. 585-595, 1999.
\bibitem{verduhan}S. Verdu and T. S. Han, "A general formula for channel capacity," in IEEE Transactions on Information Theory, vol. 40, no. 4, pp. 1147-1157, July 1994.
\bibitem{codingtheorems}I. Csiszár and J. Körner, Information Theory, "Coding Theorems for
Discrete Memoryless Systems," 2nd ed. Cambridge University Press,
2011.
\bibitem{simulationbased}D. M. Arnold, H. . -A. Loeliger, P. O. Vontobel, A. Kavcic and W. Zeng, "Simulation-Based Computation of Information Rates for Channels With Memory," in IEEE Transactions on Information Theory, vol. 52, no. 8, pp. 3498-3508, Aug. 2006.
\bibitem{optinfrate}P. Sadeghi, P. O. Vontobel and R. Shams, "Optimization of Information Rate Upper and Lower Bounds for Channels With Memory," in IEEE Transactions on Information Theory, vol. 55, no. 2, pp. 663-688, Feb. 2009.
\bibitem{feinstein}A. J. Thomasian, “Error Bounds for Continuous Channels,” in Fourth
London Symposium on Information Theory, C. Cherry, Ed. Butterworth, pp. 46–60, 1961.
\color{black}\bibitem{normposdef} R. A. Horn and C. R. Johnson, "Matrix Analysis," Cambridge University Press, ch. 7, pp. 493-495, 1985. \color{black}
\end{thebibliography}
\end{document}